\documentclass[draftcls, onecolumn, a4paper]{IEEEtran}
\usepackage{amsmath,amssymb,amsthm,mathrsfs,amsfonts,dsfont,stmaryrd, wasysym,marvosym}
\usepackage{mathtools}
\usepackage{color}
\usepackage{graphicx,subcaption,cite,epsfig,dblfloatfix,url,xcolor} 
\usepackage{etoolbox}
\allowdisplaybreaks[4] 

\newtheorem{theorem}{\textbf{Theorem}}
\newtheorem{corollary}{\textbf{Corollary}}
\newtheorem{proposition}{\textbf{Proposition}}

\newtheorem{definition}{\textbf{Definition}}
\newtheorem{lemma}{\textbf{Lemma}}

\newtheorem{remark}{\textbf{Remark}}
\newcommand{\dv}{\mathbf} % determenistic vector
\newcommand{\mc}{\mathcal} % determenistic vector
\newcommand{\mkv}{-\!\!\!\!\minuso\!\!\!\!-}
 
\newcommand{\squeezeup}{\vspace{-1em}}

\usepackage{pxfonts}
\usepackage{dsfont}
\graphicspath{{./}}

\newcommand*\xbar[1]{%
    \hbox{%
		 \vbox{%
		 \hrule height 0.5pt % The actual bar
		 \kern0.5ex%         % Distance between bar and symbo
		 \hbox{%
		 \kern-0.1em%      % Shortening on the left side
		\ensuremath{#1}%
		\kern-0.1em%      % Shortening on the right side
		}%
		}%
		}%
		} 

\begin{document}
%\fontencoding{OT1}\fontsize{9.4}{11.25pt}\selectfont
\title{Rate-Exponent Region for a Class of Distributed Hypothesis Testing Against Conditional Independence Problems\\}

\author{Abdellatif Zaidi$\:^{\dagger}$$\:^{\nmid}$\vspace{0.3cm}\\
$^{\dagger}$ Universit\'e Paris-Est, Champs-sur-Marne 77454, France\\
$^{\nmid}$ Mathematical and Algorithmic Sciences Lab., Paris Research Center, Huawei France
abdellatif.zaidi@u-pem.fr
}

% make the title area
\maketitle

\begin{abstract}
We study a class of $K$-encoder hypothesis testing against conditional independence problems. Under the criterion that stipulates minimization of the Type II error subject to a (constant) upper bound $\epsilon$ on the Type I error, we  characterize the set of encoding rates and exponent for both discrete memoryless and memoryless vector Gaussian settings. For the DM setting, we provide a converse proof and show that it is achieved using the Quantize-Bin-Test scheme of Rahman and Wagner. For the memoryless vector Gaussian setting, we develop a tight outer bound by means of a technique that relies on the de Bruijn identity and the properties of Fisher information. In particular, the result shows that for memoryless vector Gaussian sources the rate-exponent region is exhausted using the Quantize-Bin-Test scheme with \textit{Gaussian} test channels; and there is \textit{no} loss in performance caused by restricting the sensors' encoders not to employ time sharing. Furthermore, we also study a variant of the problem in which the source, not necessarily Gaussian, has finite differential entropy and the sensors' observations noises under the null hypothesis are Gaussian. For this model, our main result is an upper bound on the exponent-rate function. The bound is shown to mirror a corresponding explicit lower bound, except that the lower bound involves the source power (variance) whereas the upper bound has the source entropy power. Part of the utility of the established bound is for investigating asymptotic exponent/rates and losses incurred by distributed detection as function of the number of sensors.
\end{abstract}

%=======================================================================================
\section{Introduction}
%=======================================================================================

Consider the multiterminal detection system shown in Figure~\ref{fig-distributed-hypothesis-testing}. In this problem, a memoryless vector source $(X,Y_0,Y_1,\hdots,Y_K)$, $K \geq 1$, has a joint distribution that depends on two hypotheses, a null hypothesis $H_0$ and an alternate hypothesis $H_1$. A detector that observes directly the pair $(X,Y_0)$ but only receives summary information of the sensors' observations $(Y_1,\hdots,Y_K)$ seeks to determine which of the two hypotheses is true. Specifically, Encoder $k$, $1 \leq k \leq K$, which observes an independent and identically distributed (i.i.d.) string $Y^n_k$, sends a message $M_k$ to the detector at finite rate of $R_k$ bits per observation over a noise-free channel; and the detector makes its decision between the two hypotheses on the basis of the received messages $(M_1,\hdots,M_K)$ as well as the available pair $(X^n,Y^n_0)$. In doing so, the detector can make two types of error: Type I error (guessing $H_1$ while $H_0$ is true) and Type II error (guessing $H_0$ while $H_1$ is true). The type II error probability can decrease exponentially fast with the size $n$ of the i.i.d. strings, say with an exponent $E$; and, classically, one is interested is characterizing the set of achievable rate-exponent tuples $(R_1,\hdots,R_K,E)$ in the regime in which the probability of the Type I error is kept below a prescribed small value $\epsilon$. This problem, which was first introduced by Berger~\cite{B79} and then studied further in~\cite{AC86,H87,SP92}, arises naturally in many applications. Recent developments include analysis of the tradeoff between the two types of error exponents~\cite{WK19} or from the perspective of information spectrum~\cite{W17}, and extensions to networks with multiple sensors \cite{RW12,ZL14,ZE-A19a,Z20,UEZ20b}, multiple detectors \cite{SWT18, SW21}, interactive terminals~\cite{TC08a,XK12}, multi-hop networks~\cite{ZL14,ZL15,SWW19,EWZ18,EZW18},  noisy channels~\cite{SW20,SG20} and scenarios with privacy constraints~\cite{LSCT17,SGC18,SCG20,GBST19}. Its theoretical understanding, however, is far from complete, even from seemingly simple instances of it.

  %----------------------------------------
\begin{figure}[h!]
\centering
\includegraphics[width=0.6\linewidth]{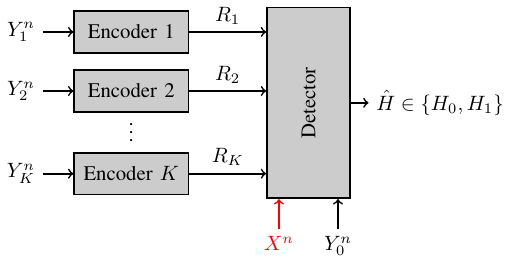}
\caption{Distributed hypothesis testing against conditional independence.} 
\squeezeup
\label{fig-distributed-hypothesis-testing}
\end{figure}
%----------------------------------------

One important such instances was studied by Rahman and Wagner in~\cite{RW12}. In~\cite{RW12}, the two hypotheses are such that $X$ and $(Y_1,\hdots,Y_K)$ are correlated conditionally given $Y_0$ under the null hypothesis $H_0$; and they are independent conditionally given $Y_0$ under the alternate hypothesis $H_1$, i.e.,~\footnote{In fact, the model of~\cite{RW12} also involves a random variable $Y_{K+1}$, which is chosen here to be deterministic as it is not relevant for the analysis and discussion that will follow in this paper.}
\begin{subequations}
\begin{align}
\label{distribution-under-null-hypothesis-special-case-ht-against-conditional-independence-DM-Rahman-Wagner}
H_0  &: P_{X,Y_0,Y_1,\hdots,Y_K} = P_{Y_0} P_{X,Y_1,\hdots,Y_K|Y_0} \\
H_1  &: Q_{X,Y_0,Y_1,\hdots,Y_K} = P_{Y_0} P_{X|Y_0} P_{Y_1,\hdots,Y_K|Y_0}. 
\label{distribution-under-alternate-hypothesis-special-case-ht-against-conditional-independence-DM-Rahman-Wagner}
\end{align}
\label{distributions-under-null-and-alternate-hypotheses-special-case-ht-against-conditional-independence-DM-Rahman-Wagner}
\end{subequations}
Note that $(Y_0,Y_1,\hdots,Y_K)$ and $(Y_0,X)$ have the same distributions under both hypotheses; and the multiterminal problem~\eqref{distributions-under-null-and-alternate-hypotheses-special-case-ht-against-conditional-independence-DM-Rahman-Wagner} is a multi-encoder version of the single-encoder test against independence studied by Ahlswede and Csisz\'ar in~\cite[Theorem 2]{AC86}. For the problem~\eqref{distributions-under-null-and-alternate-hypotheses-special-case-ht-against-conditional-independence-DM-Rahman-Wagner} Rahman and Wagner provided inner and outer bounds on the rate-exponent region which do \textit{not} match in general (see~\cite[Theorem 1]{RW12} for the inner bound and~\cite[Theorem 2]{RW12} for the outer bound). The inner bound of~\cite[Theorem 1]{RW12} is similar to a generalized  Berger-Tung inner bound for distributed source coding~\cite{B77,T78}; and is based on a scheme, named Quantize-Bin-Test (QBT) therein, in which like in the Shimokawa–Han–Amari scheme~\cite{Shimokawa1994} the encoders quantize and then bin their observations but the detector performs the test directly using the bins.

In this paper, we study a class of the distributed hypothesis testing problem~\eqref{distributions-under-null-and-alternate-hypotheses-special-case-ht-against-conditional-independence-DM-Rahman-Wagner} obtained by restricting the joint distribution of the variables under the null hypothesis $H_0$ to satisfy the Markov chain
\begin{equation}
Y_{\mc S} \mkv  (X, Y_0)  \mkv Y_{{\mc S}^c} \quad \forall \:\: \mc S \subseteq \mc K := \{1,\ldots,K\}
\label{markov-chain}
\end{equation}
i.e., the encoders' observations $\{Y_k\}_{k \in \mc K}$ are independent conditionally given $(X,Y_0)$. We investigate both discrete memoryless (DM) and memoryless vector Gaussian models. For the DM setting, we provide a converse proof and show that it is achieved using the Quantize-Bin-Test scheme of~\cite[Theorem 1]{RW12}. Our converse proof is strongly inspired by that of the rate-distortion region of the Chief-Executive Officer (CEO) problem under logarithmic loss of Courtade and Weissman~\cite[Theorem 10]{CW14}. In fact, with an easy entropy characterization of the rate-exponent region that we develop here the problem is shown equivalent operationally to an CEO problem in which the remote source is $X$, agent $k$ observes $Y_k$, the decoder observes side information (SI) $Y_0$ and wants to reconstruct the remote source $X$ to within average distortion level $(H(X|Y_0)-E)$, and where the distortion is measured under logarithmic loss. It appears that the result of our converse can be implied by Rahman-Wagner outer bound of~\cite[Theorem 2]{RW12} when in the problem~\eqref{distributions-under-null-and-alternate-hypotheses-special-case-ht-against-conditional-independence-DM-Rahman-Wagner} one imposes the Markov condition~\eqref{markov-chain} on the distribution under the null hypothesis. This, moreover, also means that for the multiterminal CEO problem under logarithmic loss of~\cite{CW14} the outer bound of Wagner-Anantharam of~\cite[Theorem 1]{Wagner2008b} implies the converse part of their Theorem 10 therein. Finally, we note that, for general distributions under the null hypothesis, i.e., without the Markov chain~\eqref{markov-chain}, prior to this work the optimality of the Quantize-Bin-Test scheme of~\cite{RW12} for the problem of testing against conditional independence was known only for the special case of a single encoder, i.e., $K=1$~\cite[Theorem 3]{RW12}, a result which can also be recovered from our result in this paper. 

For the vector Gaussian setting we provide an explicit characterization of the rate-exponent region. For the proof of the converse part of this result, essentially we develop an outer bound by means of a technique that relies on the de Bruijn identity and the properties of Fisher information; and we show that it is tight. Past application of these techniques was shown recently to yield the optimal region for the related vector Gaussian CEO problem under logarithmic loss in~\cite{UEZ20b}, while previously found generally non-tight for the classic squared error distortion measure~\cite{EU14}. In particular, our result here shows that for memoryless vector Gaussian sources the rate-exponent region is exhausted using the Quantize-Bin-Test Scheme of~\cite[Theorem 1]{RW12} with \textit{Gaussian} test channels. Furthermore, it also shows that there is no loss in performance caused by restricting the sensors' encoders \textit{not} to employ time sharing. This provides what appears to be the first optimality result for the Gaussian hypothesis testing against conditional independence problem in the vector sources case.

Furthermore, we broaden our view to also study a generalization of the $K$-encoder scalar Gaussian hypothesis testing against independence problem in which the sensors' observations under the null hypothesis are independent noisy versions of $X$, with Gaussian noises, but $X$ itself is an arbitrary continuous memoryless source. For instance, the distribution of $X$, not necessarily Gaussian, is arbitrary and has non-zero finite entropy power. We recall that the entropy power of a continuous random variable $X$ which has density $p_X(x)$ is defined as
\begin{equation}
N(X) = \frac{e^{2h(X)}}{2{\pi}e}
\label{definition-entropy-power}
\end{equation}
where $h(X)$ denotes the differential entropy of $X$. In this case, we establish an upper bound on the exponent rate function. It is shown that the bound exactly mirrors a corresponding \textit{explicit} lower bound, except that the lower bound has the source power (variance) whereas the upper bound has the source entropy power. The bounds do not depend on auxiliaries; and, while they hold generally for arbitrary distributions of source $X$ with finite differential entropy, their utility is mostly in that they reflect the right behavior as a function of the number of sensors.

\vspace{-0.2cm}

%---------------------------------------------------------------------------------------
\subsection{Outline and Notation}
%---------------------------------------------------------------------------------------

The rest of this paper is organized as follows. Section~\ref{secII} provides a formal description of the hypothesis testing problem that we study in this paper, as well as some definitions that are related to it.  Sections~\ref{secIII} and ~\ref{secIV} contain the main results of this paper. Section~\ref{secIII} provides a single-letter characterization of the rate-exponent region in the DM setting, as well as an explicit characterization of the region for the case of memoryless vector Gaussian sources. Section~\ref{secIV} provides an upper bound on the exponent-rate function for the case in which the sensors' noises are Gaussian but the source itself is memoryless continuous with arbitrary density that has finite differential entropy. This section also contains application to the study of asymptotics of the exponent-rate function for a large number of sensors. The proofs are deferred to the appendices section.

Throughout this paper, we use the following notation. Upper case letters are used to denote random variables, e.g., $X$; lower case letters are used to denote realizations of random variables, e.g., $x$; and calligraphic letters denote sets, e.g., $\mc X$.  The cardinality of a set $\mc X$ is denoted by $|\mc X|$. The closure of a set $\mc A$ is denoted by $\xbar{\mc A}$. The length-$n$ sequence $(X_1,\ldots,X_n)$ is denoted as  $X^n$; and, when confusion is not possible, for integers $j$ and $k$ such that $1 \leq k \leq j \leq n$ the sub-sequence $(X_k,X_{k+1},\ldots, X_j)$ is denoted as  $X_{k}^j$. Probability mass functions (pmfs) are denoted by $P_X(x)=\mathrm{Pr}\{X=x\}$; and, sometimes, for short, as $p(x)$. We use $\mc P(\mc X)$ to denote the set of discrete probability distributions on $\mc X$.  Boldface upper case letters denote vectors or matrices, e.g., $\dv X$, where context should make the distinction clear. For an integer $K \geq 1$, we denote the set of integers smaller or equal $K$ as $\mc K = \{ k \in \mathbb{N} \: : \: 1 \leq k \leq K\}$. For a set of integers $\mc S \subseteq \mc K$, the complementary set of $\mc S$ is denoted by $\mc S^c$, i.e., $\mc S^c = \{k \in \mathbb{N} \: : \: k \in \mc K \setminus \mc S\}$. Sometimes, for convenience we will need to define $\bar{\mc S}$ as $\bar{\mc S} = \{0\} \cup \mc S^c$. For a set of integers $\mc S \subseteq \mc K$; the notation $X_{\mc S}$ designates the set of random variables $\{X_k\}$ with indices in the set $\mc S$, i.e., $X_{\mc S}=\{X_k\}_{k \in \mc S}$. We denote the covariance of a zero mean, complex-valued, vector $\dv X$ by $\mathbf{\Sigma}_{\mathbf{x}} =\mathbb{E}[\mathbf{XX}^{\dag}]$, where $(\cdot)^{\dag}$ indicates conjugate transpose. Similarly, we denote the cross-correlation of two zero-mean vectors $\dv X$  and $\dv Y$ as $\mathbf{\Sigma}_{\mathbf{x},\mathbf{y}} = \mathbb{E}[\mathbf{XY}^{\dag}]$, and the conditional correlation matrix of $\mathbf{X}$ given $\mathbf{Y}$ as $\mathbf{\Sigma}_{\mathbf{x}|\mathbf{y}} = \mathbb{E}\big[\big(\dv X - \mathbb{E}[\dv X|\dv Y]\big)\big(\dv X - \mathbb{E}[\dv X|\dv Y]\big)^{\dag}\big]$ i.e., $\mathbf{\Sigma}_{\mathbf{x}|\mathbf{y}} = \mathbf{\Sigma}_{\mathbf{x}}-\mathbf{\Sigma}_{\mathbf{x},\mathbf{y}}\mathbf{\Sigma}_{\mathbf{y}}^{-1}\mathbf{\Sigma}_{\mathbf{y},\mathbf{x}}$. For matrices $\dv A$ and $\dv B$, the notation $\mathrm{diag}(\dv A, \dv B)$ denotes the block diagonal matrix whose diagonal elements are the matrices $\dv A$ and $\dv B$ and its off-diagonal elements are the all zero matrices. Also, for a set of integers $\mc J \subset \mathbb{N}$ and a family of matrices $\{\dv A_i\}_{i \in \mc J}$ of the same size, the notation $\dv A_{\mc J}$ is used to denote the (super) matrix obtained by concatenating vertically the matrices $\{\dv A_i\}_{i \in \mc J}$, where the indices are sorted in the ascending order, e.g, $\dv A_{\{0,2\}}=[\dv A^{\dag}_0, \dv A^{\dag}_2]^{\dag}$.

% \newpage

%=======================================================================================
\section{Problem Formulation}~\label{secII}
%=======================================================================================

\vspace{-0.4cm}

Consider a $(K+2)$-dimensional memoryless source $(X,Y_0,Y_1,\ldots,Y_K)$ with finite alphabet $\mc X \times \mc Y_0 \times \mc Y_1 \times \ldots \times \mc Y_K$. The joint probability mass function (pmf) of $(X,Y_0,Y_1,\ldots,Y_K)$ is assumed to be determined by a hypothesis $H$  that takes one of two values, a null hypothesis $H_0$ and an alternate hypothesis $H_1$. Under the null hypothesis $H_0$, it is assumed that $X$ and $(Y_0,Y_1,\hdots,Y_K)$ are correlated and the joint distribution of $(X,Y_0,Y_1,\hdots,Y_K)$ satisfies the following Markov chain 
\begin{equation}
Y_{\mc S} \mkv  (X, Y_0)  \mkv Y_{{\mc S}^c} \quad \forall \:\: \mc S \subseteq \mc K := \{1,\ldots,K\}.
\label{eq:MKChain_pmf}
\end{equation}
Under the alternate hypothesis $H_1$, it is assumed that $X$ and $(Y_1,\hdots,Y_K)$ are independent conditionally given $Y_0$. That is, 
\begin{subequations}
\begin{align}
\label{distribution-under-null-hypothesis-ht-against-conditional-independence-DM}
& H_0 : P_{X,Y_0,Y_1\hdots,Y_K} = P_{X,Y_0} \prod_{k=1}^K P_{Y_k|X,Y_0} \\
&H_1 : Q_{X,Y_0,Y_1\hdots,Y_K} = Q_{Y_0} Q_{X|Y_0} Q_{Y_1,\hdots,Y_K|Y_0}. 
\label{distribution-under-alternate-hypothesis-ht-against-conditional-independence-DM}
\end{align}
\label{distribution-under-null-and-alternate-hypotheses-ht-against-conditional-independence-DM}
\end{subequations}

\noindent Throughout we make the assumption that the distributions $P$ and $Q$ have same $(X,Y_0)$- and $(Y_0,Y_1,\hdots,Y_K)$-marginals, i.e.,
\begin{equation}
P_{X,Y_0} = Q_{X,Y_0} \qquad \text{and} \qquad  P_{Y_0,Y_1,\hdots,Y_K} = Q_{Y_0,Y_1,\hdots,Y_K}.
\label{condition-on-same-marginals-under-H0-and-H1}
\end{equation}

\noindent Let now $\{(X_i,Y_{0,i},Y_{1,i},\ldots,Y_{K,i})\}^n_{i=1}$ be a sequence of $n$ independent copies of $(X,Y_0,Y_1,\ldots,Y_K)$; and consider the detection system shown in Figure~\ref{fig-distributed-hypothesis-testing}. Here, there are $K$ sensors and one detector. Sensor $k \in \mc K$ observes the memoryless source component $Y^n_k$ and sends a message $M_k = {\phi}^{(n)}_k(Y^n_k)$ to the detector, where the mapping
\begin{equation}
{\phi}^{(n)}_k \: : \:  \mc Y^n_k \rightarrow \{1,\ldots,M^{(n)}_k\}
\end{equation}
designates the encoding operation at this sensor. The detector observes the pair $(X^n,Y^n_0)$ and uses them, as well as the messages $\{M_1,\hdots,M_K\}$ gotten from the sensors, to make a decision between the two hypotheses, based on a decision rule
\begin{equation}
{\psi}^{(n)} \: : \{1,\ldots,M^{(n)}_1\}\times \ldots \times  \{1,\ldots,M^{(n)}_K\}\times \mc X^n \times \mc Y_0^n  \rightarrow \{H_0,H_1\}.
\label{decision-rule-hypothesis-testing}
\end{equation}
The mapping~\eqref{decision-rule-hypothesis-testing} is such that ${\psi}^{(n)}(m_1,\hdots,m_K,x^n,y^n_0)=H_0$ if $(m_1,\hdots,m_K,x^n,y^n_0) \in \mc A_n$ and $H_1$ otherwise, with
\begin{equation*}
\mc A_n \subseteq \prod_{k=1}^n \{1,\ldots,M^{(n)}_k\} \times \mc X^n \times \mc Y_0^n
\end{equation*}
designating the acceptance region for $H_0$. The encoders $\{{\phi}^{(n)}_k\}_{k=1}^K$ and the detector ${\psi}^{(n)}$ are such that the Type I error probability does not exceed a prescribed level $\epsilon \in [0,1]$, i.e.,
\begin{equation}
P_{{\phi}^{(n)}_1(Y^n_1),\hdots,{\phi}^{(n)}_K(Y^n_K),X^n,Y^n_0} (\mc A^c_n) \leq \epsilon
\end{equation}
and the Type II error probability does not exceed $\beta$, i.e.,
\begin{equation}
Q_{{\phi}^{(n)}_1(Y^n_1),\hdots,{\phi}^{(n)}_K(Y^n_K),X^n,Y^n_0} (\mc A_n) \leq \beta.
\end{equation}
\begin{definition}
A rate-exponent tuple $(R_1,\hdots,R_K,E)$ is achievable for a fixed $\epsilon \in [0,1]$ if for any positive $\delta$ and sufficiently large $n$ there exist encoders $\{{\phi}^{(n)}_k\}_{k=1}^K$ and a detector ${\psi}^{(n)}$ such that
\begin{subequations}
\begin{align}
\frac{1}{n} \log M^{(n)}_k &\leq R_k + \delta \:\: \text{for all}\:\: k \in \mc K, \:\: \text{and} \\
 -\frac{1}{n} \log \beta &\geq E - \delta. 
\end{align}
\end{subequations}
The rate-exponent region $\mc R_{\text{HT}}$ is defined as
\begin{equation}
\mc R_{\text{HT}} := \bigcap_{\epsilon > 0} \mc R_{\text{HT},\epsilon},
\label{definition-rate-exponent-region}
\end{equation}
where $\mc R_{\text{HT},\epsilon}$ is the set of all achievable rate-exponent vectors for a fixed $\epsilon \in [0,1]$. 
\qed
\end{definition}

%=======================================================================================
\section{Rate-Exponent Results}~\label{secIII}
%=======================================================================================

\vspace{-0.7cm}

%---------------------------------------------------------------------------------------
\subsection{Discrete Memoryless Case}~\label{secIII_subsecA}
%---------------------------------------------------------------------------------------

\vspace{-0.2cm}

 We start with an entropy characterization of the rate-exponent region $\mc R_{\text{HT}}$ as defined by~\eqref{definition-rate-exponent-region}. Let
\begin{equation}
\mc R^{\star} = \bigcup_{n} \bigcup_{\{{\phi}^{(n)}_k\}_{k \in \mc K}}  \mc R^{\star}\left(n,\{{\phi}^{(n)}_k\}_{k \in \mc K}\right)
\end{equation}
where
\begin{subequations}
\begin{align}
\label{entropy-characterization-rate-exponent-region-distributed-ht-against-conditional-independence-rates-inequalities}
\mc R^{\star}\left(n,  \{{\phi}^{(n)}_k\}_{k \in \mc K}\right) =  \Big\{(&R_1,\hdots,R_K,E) \:\: \text{s.t.} \nonumber\\
& R_k \geq \frac{1}{n} \log|{\phi}^{(n)}_k(Y^n_k)| \:\: \text{for all}\:\: k \in \mc K, \:\: \text{and} \\
& E \leq \frac{1}{n} I(\{{\phi}^{(n)}_k(Y^n_k)\}_{k \in \mc K};X^n|Y^n_0) \Big\}.
\label{entropy-characterization-rate-exponent-region-distributed-ht-against-conditional-independence-exponent-inequality}
\end{align}
\label{entropy-characterization-rate-exponent-region-distributed-ht-against-conditional-independence}
\end{subequations}

\noindent We have the following proposition the proof of which is essentially similar to that of~\cite[Theorem 1]{AC86} and appears in Appendix~\ref{appendix-proposition-entropy-characterization-rate-exponent-region-distributed-ht-against-conditional-independence}.

\begin{proposition}~\label{proposition-entropy-characterization-rate-exponent-region-distributed-ht-against-conditional-independence}
$\mc R_{\text{HT}} = \xbar{\mc R^{\star}}$.
\end{proposition}

\noindent The result of Proposition~\ref{proposition-entropy-characterization-rate-exponent-region-distributed-ht-against-conditional-independence} essentially means that the studied hypothesis testing problem is operationally equivalent to a chief executive officer (CEO) source coding problem where the distortion is measured under logarithmic loss. Specifically, this equivalent CEO problem is one in which the remote source is $X$; there are $K$ agents observing noisy versions of it, with agent $k$ observing $Y_k$; and the decoder observes side information (SI) $Y_0$ and wants to reconstruct the remote source $X$ to within average logarithmic loss distortion $(H(X|Y_0)-E)$. The latter problem was solved in~\cite[Theorem 10]{CW14} in the case of no decoder SI  (i.e., $Y_0=\emptyset$) but its proof carries over with minimal changes to the case in which the decoder is equipped with SI $Y_0$. Thus, with the result of Proposition~\ref{proposition-entropy-characterization-rate-exponent-region-distributed-ht-against-conditional-independence} and a rather straightforward generalization of~\cite[Theorem 10]{CW14} we have the following theorem which provides a single-letter characterization of the rate-exponent region $\mc R_{\text{HT}}$.

\begin{theorem}~\label{theorem-rate-exponent-region-hypothesis-testing-DM-case}
The rate-exponent region $\mc R_{\text{HT}}$ is given by the union of all  non-negative tuples $(R_1,\ldots,R_K,E)$ that satisfy, for all subsets $\mc S \subseteq \mc K$, 
\begin{equation}
E \leq I(U_{\mc S^c};X|Y_0,Q) + \sum_{k \in \mc S} \big(R_k - I(Y_k;U_k|X,Y_0,Q)\big)
\end{equation}
for some auxiliary random variables $(U_1,\ldots,U_K,Q)$ with distribution $P_{U_{\mc K},Q}$ such that 
\begin{equation}
P_{X, Y_0, Y_{\mc K}, U_{\mc K}, Q} = P_Q  P_{X, Y_0} \prod_{k=1}^K P_{Y_k|X, Y_0} \: \prod_{k=1}^{K} P_{U_k|Y_k,Q}.
\label{joint-measure-theorem-rate-exponent-region-hypothesis-testing-DM-case}
\end{equation}
\end{theorem}
A direct proof of the achievability part of Theorem~\ref{theorem-rate-exponent-region-hypothesis-testing-DM-case} follows by an easy application of the Quantize-Bin-Test scheme of Rahman and Wagner~\cite[Theorem 1]{RW12}. The interested reader may also find an alternate, direct, proof of its converse part in Appendix~\ref{appendix-proof-theorem-rate-exponent-region-hypothesis-testing-DM-case}.

Comparatively, the hypothesis testing model of~\cite{RW12} is one in which under the null hypothesis $(Y_1,\hdots,Y_K)$ are arbitrarily correlated among them and with the pair $(X,Y_0)$; and under the alternate hypothesis $Y_0$ induces conditional independence between $(Y_1,\hdots,Y_K)$ and $X$. More precisely, the joint distributions of $(X,Y_0,Y_1,\hdots,Y_K)$ under the null and alternate hypotheses as considered in~\cite{RW12} are
\begin{subequations}
\begin{align}
\label{distribution-under-null-hypothesis-ht-against-conditional-independence-DM-Rahman-Wagner}
H_0  &: \tilde{P}_{X,Y_0,Y_1\hdots,Y_K} = P_{Y_0} P_{X|Y_0} P_{Y_1,\hdots,Y_K|X,Y_0} \\
H_1  &: \tilde{Q}_{X,Y_0,Y_1\hdots,Y_K} = P_{Y_0} P_{X|Y_0} P_{Y_1,\hdots,Y_K|Y_0}. 
\label{distribution-under-alternate-hypothesis-ht-against-conditional-independence-DM-Rahman-Wagner}
\end{align}
\label{distributions-under-null-and-alternate-hypotheses-ht-against-conditional-independence-DM-Rahman-Wagner}
\end{subequations}
 For this more general model, they provide inner and outer bounds on the rate-exponent region which do \textit{not} match in general (see~\cite[Theorem 1]{RW12} for the inner bound and~\cite[Theorem 2]{RW12} for the outer bound). Our Theorem~\ref{theorem-rate-exponent-region-hypothesis-testing-DM-case} shows that if, in addition, the joint distribution of the variables under the null hypothesis $H_0$ is restricted to satisfy the Markov chain condition~\eqref{eq:MKChain_pmf}, then the Quantize-Bin-Test scheme of~\cite[Theorem 1]{RW12} is optimal. Accordingly, the reader may wonder whether the converse of Theorem~\ref{theorem-rate-exponent-region-hypothesis-testing-DM-case} could be implied by Rahman-Wagner outer bound of~\cite[Theorem 2]{RW12} when specialized to the test setting studied here. The answer to this question, brought to the attention of the author during the revision of this paper, appears to be affirmative. To see this, recall that the outer bound of~\cite[Theorem 2]{RW12}, denoted hereafter as $\mc R^{\text{out}}_{\text{RW}}$, is given by
\begin{equation}
\mc R^{\text{out}}_{\text{RW}} = \bigcap_{A \in \mc A} \bigcup_{\lambda_0 \in \Lambda_0} \mc R^{\text{out}}_{\text{RW}}(A,\lambda_0)
\label{response-reviewer2-step8}
\end{equation}
where:
\begin{itemize}
\item[i)] $\mc A$ is the set of finite-alphabet random variable $A$ such that $Y_1,\hdots,Y_K,X$ are conditionally independent given $(A,Y_0)$;
\item[ii)] $\Lambda_0$ is the set of finite-alphabet random variables $\lambda_0=(U_1,\hdots,U_K, W,Q)$ such that: 
\begin{itemize}
\item[(a)] $(W,Q)$ is independent of $(Y_1,\hdots,Y_K,X,Y_0)$
\item[(b)] $U_k \mkv (Y_k,W,Q) \mkv (U_{k^c}, Y_{k^c}, X, Y_0)$ for all $k \in \mc K$;
\end{itemize}
\item[iii)] for given $A \in \mc A$ and $\lambda_0 \in \Lambda_0$ for which the joint distribution of $A$, $(X,Y_0,Y_1,\hdots,Y_K)$ and $\lambda_0$ satisfies the Markov chain condition
\begin{equation}
A \mkv (Y_1,\hdots,Y_K,X,Y_0) \mkv \lambda_0
\end{equation} 
and $\mc R^{\text{out}}_{\text{RW}}(A,\lambda_0)$ is defined as the set of all non-negative $(R_1,\hdots,R_K,E)$ for which
\begin{subequations}
\begin{align}
\label{rahman-wagner-outer-bound-constraint-on-rates}
\sum_{k \in \mc S} R_k & \geq I(\dv U_{\mc S};A|\dv U_{\mc S^c},Y_0,Q) + \sum_{k \in \mc S} I(U_k;Y_k|A,W,Y_0,Q), \:\: \forall \mc S \subseteq \mc K\\
E &\leq I(U_1,\hdots,U_K;X|Y_0,Q).
\label{rahman-wagner-outer-bound-constraint-on-exponent}
\end{align}
\label{rahman-wagner-outer-bound}
\end{subequations}
\end{itemize}

\noindent Let $ (U_1,\hdots,U_K, W,Q) \in \Lambda_0$. Noticing that $X \in \mc A$ and setting $A= X$, the inequality~\eqref{rahman-wagner-outer-bound-constraint-on-rates} can be weakened as 
\begin{align}
\label{bound-rahman-wagner-first-step}
\sum_{k \in \mc S} R_k & \geq I(\dv U_{\mc S};X|\dv U_{\mc S^c},Y_0,Q) + \sum_{k \in \mc S} I(U_k;Y_k|X,W,Y_0,Q) \\
& = I(\dv U_{\mc K};X|Y_0,Q) - I(\dv U_{\mc S^c};X|Y_0,Q) + \sum_{k \in \mc S} I(U_k;Y_k|X,W,Y_0,Q) \\
&\geq E - I(\dv U_{\mc S^c};X|Y_0,Q) + \sum_{k \in \mc S} I(U_k;Y_k|X,W,Y_0,Q)
\label{rahman-wagner-outer-bound-specialized-to-markov-chain-of-null-hypothesis}
\end{align}
where the last inequality follows by using~\eqref{rahman-wagner-outer-bound-constraint-on-exponent}. Also, we have
\begin{align}
I(\dv U_{\mc S^c};X|Y_0,Q) &= H(X|Y_0,Q) - H(X|\dv U_{\mc S^c}, Y_0, Q) \\
&\stackrel{(a)}{=} H(X|Y_0,W,Q) - H(X|\dv U_{\mc S^c}, Y_0, Q) \\
&\stackrel{(b)}{\leq} H(X|Y_0,W,Q) - H(X|\dv U_{\mc S^c}, Y_0, W, Q) \\
& = I(\dv U_{\mc S^c};X|Y_0,W,Q)
\label{bound-second-term}
\end{align}
where $(a)$ holds since $(W,Q)$ is independent of $(X,Y_0)$ and $(b)$ holds since conditioning reduces entropy.
 
\noindent Combining~\eqref{rahman-wagner-outer-bound-specialized-to-markov-chain-of-null-hypothesis} and~\eqref{bound-second-term}, we get that for all $\mc S \subseteq K$ we have
Thus, we have the bound 
\begin{equation}
\sum_{k \in \mc S} R_k  \geq E - I(\dv U_{\mc S^c};X|Y_0,W,Q) + \sum_{k \in \mc S} I(U_k;Y_k|X,W,Y_0,Q).
\label{bound-rahman-wagner-last-step}
\end{equation}
Thus, the variable $W$ can be absorbed into the time sharing random variable $Q$, and one gets the expression of the above Theorem~\ref{theorem-rate-exponent-region-hypothesis-testing-DM-case}. 

\begin{remark}\label{relation-to-Courtade-Weissman}
 For reasons that are essentially similar to the above it is not difficult to see that, for the related $K$-encoder CEO problem under logarithmic loss, $k \geq 2$,  the outer bound of Wagner-Anantharam of~\cite[Theorem 1]{Wagner2008b} implies the converse part of Courtade-Weissman~\cite[Theorem 10]{CW14}. 
\end{remark}

\begin{remark}\label{remark-optimality-QBT-DM-case}
Prior to this work, the optimality of the QBT scheme of~\cite{RW12} for the problem of testing against conditional independence was known only for the special case of a single encoder, i.e., $K=1$~\cite[Theorem 3]{RW12}, a result which can also be recovered from Theorem~\ref{theorem-rate-exponent-region-hypothesis-testing-DM-case}. 
\end{remark}

\vspace{-0.4cm}

%---------------------------------------------------------------------------------------
\subsection{Memoryless Vector Gaussian Case}~\label{secIII_subsecB}
%---------------------------------------------------------------------------------------

\vspace{-0.4cm}

We now turn to a continuous example of the hypothesis testing problem studied in this paper. Here,  $(\dv X,\dv Y_0,\dv Y_1,\hdots,\dv Y_K)$ is a zero-mean circularly-symmetric complex-valued Gaussian random vector. Without loss of generality, let 
\begin{equation}~\label{mimo-gaussian-ht-model-2}
\dv Y_0 = \dv H_0 \dv X  + \dv Z_0
\end{equation}
 where $\dv H_0 \in \mathds{C}^{n_0\times n_x}$, $\dv X \in \mathds{C}^{n_x}$ and $\dv Z_0 \in \mathds{C}^{n_0}$ are independent Gaussian vectors with zero-mean and covariance matrices $\dv\Sigma_{\dv x} \succ \dv 0$ and $\dv\Sigma_0 \succ \dv 0$, respectively. The vectors $(\dv Y_1, \hdots,\dv Y_K)$ and $\dv X$ are correlated under the null hypothesis $H_0$ and are independent under the alternate hypothesis $H_1$. Specifically, under the null hypothesis 
\begin{equation}
H_0 \: : \dv Y_k = \dv H_k \dv X + \dv Z_k, \quad\text{for all}\:\: k \in \mc K
\label{distributions-under-null-hypothesis-ht-against-conditional-independence-vector-Gaussian}
\end{equation}
where the noise vectors $(\dv Z_1,\hdots,\dv Z_K)$ are jointly Gaussian with zero mean and covariance matrix $\dv\Sigma_{\dv n_{\mc K}}  \succ \dv 0$, and assumed to be independent from $\dv X$ but correlated among them and with $\dv Z_0$, such that for every $\mc S \subseteq \mc K$, 
\begin{equation}  
\dv Z_{\mc S} \mkv  \dv Z_0  \mkv \dv Z_{\mc S^c}. 
\label{markov-chain-assumption-gaussian-hypothesis-testing-model}
\end{equation}
\noindent For every $k \in \mc K$ we denote by $\dv\Sigma_k$ the \textit{conditional} covariance matrix of noise $\dv Z_k$ conditionally given $\dv Z_0$. Under the alternate hypothesis $H_1$, the joint distribution of $(\dv X, \dv Y_0, \dv Y_1,\hdots,\dv Y_K)$, denoted as $Q_{\dv X, \dv Y_0, \dv Y_1,\hdots,\dv Y_K}$, factorizes as
\begin{equation}
H_1 \: :  Q_{\dv X, \dv Y_0,\dv Y_1,\hdots,\dv Y_K} = Q_{\dv Y_0}Q_{\dv X|\dv Y_0} Q_{\dv Y_1,\hdots,\dv Y_K |\dv Y_0}.
\label{distributions-under-alternate-hypothesis-ht-against-conditional-independence-vector-Gaussian}
\end{equation}
Here $Q_{\dv X, \dv Y_0}=P_{\dv X, \dv Y_0}$ where $P_{\dv X, \dv Y_0}$ is the joint distribution of the vector $(\dv X, \dv Y_0)$ under $H_0$ as induced by~\eqref{mimo-gaussian-ht-model-2} and $Q_{\dv Y_0, \dv Y_1,\hdots,\dv Y_K}=P_{\dv Y_0, \dv Y_1,\hdots,\dv Y_K}$ where $P_{\dv Y_0, \dv Y_1,\hdots,\dv Y_K}$ is the joint distribution of the vector $(\dv Y_0, \dv Y_1,\hdots,\dv Y_K)$ under $H_0$ as induced by~\eqref{mimo-gaussian-ht-model-2},~\eqref{distributions-under-null-hypothesis-ht-against-conditional-independence-vector-Gaussian} and~\eqref{markov-chain-assumption-gaussian-hypothesis-testing-model}.

\noindent Let $\mc R_{\text{VG-HT}}$ denote the rate-exponent region of this vector Gaussian hypothesis testing against conditional independence problem. 

\noindent For convenience, we now introduce the following notation which will be instrumental in what follows. Let, for every set $\mc S \subseteq \mc K$, the set $\bar{\mc S} = \{0\} \cup \mc S^c$. Also, for $\mc S \subseteq \mc K$ and given matrices $\{\dv\Omega_k\}_{k=1}^K$ such that $\dv 0 \preceq \dv\Omega_k \preceq \dv\Sigma_k^{-1}$, let $\boldsymbol{\Lambda}_{\bar{\mc S}}$ designate the block-diagonal matrix given by
\begin{align}~\label{equation-definition-T}
\boldsymbol{\Lambda}_{\bar{\mc S}} :=
\begin{bmatrix}
\dv 0 & \dv 0 \\
\dv 0 & \mathrm{diag}\left(\left\{ \dv\Sigma_k - \dv\Sigma_k \dv\Omega_k \dv\Sigma_k \right\}_{k\in\mc S^c}\right)
\end{bmatrix} 
\end{align}
where $\dv 0$ in the principal diagonal elements is the $n_0{\times}n_0$-all zero matrix.

\noindent The following theorem provides an explicit characterization of $\mc R_{\text{VG-HT}}$.  

\begin{theorem}~\label{theorem-rate-exponent-region-gaussian-hypothesis-testing-against-conditional-independence}
The rate-exponent region $\mc R_{\text{VG-HT}}$ of the vector Gaussian hypothesis testing against conditional independence problem is given by the set of all non-negative tuples $(R_1,\ldots, R_K,E)$ that satisfy, for all subsets $\mc S \subseteq \mc K$,  
\begin{align}
E &\leq \sum_{k \in \mc S} \big(R_k + \log \left| \dv I - \dv\Omega_k \dv\Sigma_k \right| \big) - \log \left|\dv I + \dv\Sigma_{\dv x}\dv H_0^{\dagger}\dv\Sigma_0^{-1}\dv H_0 \right| \nonumber\\
&\vspace{0.2cm} + \log \left| \dv I + \dv\Sigma_{\dv x} \dv H_{\bar{\mc S}}^\dagger \dv\Sigma_{\dv n_{\bar{\mc S}}}^{-1}\big( \dv I - \boldsymbol{\Lambda}_{\bar{\mc S}} \dv\Sigma_{\dv n_{\bar{\mc S}}}^{-1}\big)\dv H_{\bar{\mc S}}\right| 
\label{rate-exponent-region-gaussian-hypothesis-testing-against-conditional-independence}
\end{align}
for matrices $\{\dv\Omega_k\}_{k=1}^K$ such that $\dv 0 \preceq \dv\Omega_k \preceq \dv\Sigma_k^{-1}$, where $\bar{\mc S}=\{0\} \cup \mc S^c$ and $\boldsymbol{\Lambda}_{\bar{\mc S}}$ is given by~\eqref{equation-definition-T}. \qed
\end{theorem}

\begin{proof}
The proof of Theorem~\ref{theorem-rate-exponent-region-gaussian-hypothesis-testing-against-conditional-independence} appears in Appendix~\ref{appendix-proof-theorem-rate-exponent-region-gaussian-hypothesis-testing-against-conditional-independence}.
% The proof of Theorem~\ref{theorem-rate-exponent-region-gaussian-hypothesis-testing-against-conditional-independence} is given in Section~\ref{secV_subsecB}. 
\end{proof}

The direct part of Theorem~\ref{theorem-rate-exponent-region-gaussian-hypothesis-testing-against-conditional-independence} is obtained by evaluating the region of Theorem~\ref{theorem-rate-exponent-region-hypothesis-testing-DM-case}, which can be shown easily to extend to the continuous alphabet case through standard discretization arguments, using Gaussian test channels and no-time sharing. Specifically, we let $Q=\emptyset$ and $P_{U_k|\dv Y_k, Q}(u_k|\dv y_k, q)= \mc{CN}(\dv y_k, [(\dv I - \dv\Omega_k\dv\Sigma_k)^{-1}-\dv I]^{-1} \dv\Sigma_k)$. The main contribution of Theorem~\ref{theorem-rate-exponent-region-gaussian-hypothesis-testing-against-conditional-independence} is its converse part, the proof of which uses a technique that relies on the de Bruijn identity and the properties of Fisher information. The bound is similar to the outer bound for the vector Gaussian CEO problem under logarithmic loss given by Ugur \textit{et al.} in~\cite{UEZ20b}.  In particular, the result of Theorem~\ref{theorem-rate-exponent-region-gaussian-hypothesis-testing-against-conditional-independence} shows that there is no loss in performance if one restricts the auxiliaries (test channels) of the Quantize-Test-Bin scheme to be \textit{Gaussian}. Furthermore, there is \textit{no} loss in performance caused by restricting the encoders not to employ time sharing.

In the rest of this section, we elaborate on two special cases of Theorem~\ref{theorem-rate-exponent-region-gaussian-hypothesis-testing-against-conditional-independence}, the one-encoder vector Gaussian testing against conditional independence problem (i.e., $K=1$) and the $K$-encoder scalar Gaussian testing against independence problem.

%.......................................................................................
\subsubsection{The one-encoder vector Gaussian HT problem against conditional independence}~\label{secIII_subsecB_subsubsection1}
%.......................................................................................
Set $K=1$ in~\eqref{distributions-under-null-hypothesis-ht-against-conditional-independence-vector-Gaussian},~\eqref{markov-chain-assumption-gaussian-hypothesis-testing-model} and~\eqref{distributions-under-alternate-hypothesis-ht-against-conditional-independence-vector-Gaussian}. In this case the Markov chain~\eqref{markov-chain-assumption-gaussian-hypothesis-testing-model} is non-restrictive as it is trivially satisfied for all arbitrarily correlated noise at the sensor and side information $\dv Y_0$ at the detector. Theorem~\ref{theorem-rate-exponent-region-gaussian-hypothesis-testing-against-conditional-independence} then provides a complete solution of the (general) one-encoder vector Gaussian testing against conditional independence problem. The result is stated in the following Corollary.

\begin{corollary}
For the one-encoder vector Gaussian HT against conditional independence problem, the rate-exponent region is given by the set of all non-negative pairs $(R_1,E)$ that satisfy
\begin{subequations}
\begin{align}
E & \leq R_1 + \log \left|\dv I - \dv\Omega_1 \dv\Sigma_1 \right| \\
E &\leq \log \left| \dv I + \dv\Sigma_{\dv x} \dv H_{\{0,1\}}^\dagger \dv\Sigma_{\dv n_{\{0,1\}}}^{-1}\big( \dv I - \boldsymbol{\Lambda}_{\{0,1\}} \dv\Sigma_{\dv n_{\{0,1\}}}^{-1}\big)\dv H_{\{0,1\}} \right| - \log \left|\dv I + \dv\Sigma_{\dv x}\dv H_0^{\dagger}\dv\Sigma_0^{-1}\dv H_0 \right|,
\end{align}
\label{rate-exponent-region-one-encoder-testing-against-conditional-independence-vector-Gaussian-case}
\end{subequations}
for some $n_1{\times}n_1$ matrix $\dv\Omega_1$ such that $\dv 0 \preceq \dv\Omega_1 \preceq \dv\Sigma_1^{-1}$, where $\dv H_{\{0,1\}}=[\dv H^{\dag}_0, \dv H^{\dag}_1]^{\dag}$, $\dv\Sigma_{\dv n_{\{0,1\}}}$ is the covariance matrix of noise $(\dv Z_0, \dv Z_1)$ and 
\begin{align}
\boldsymbol{\Lambda}_{\{0,1\}} :=
\begin{bmatrix}
\dv 0 & \dv 0 \\
\dv 0 & \dv\Sigma_1 - \dv\Sigma_1 \dv\Omega_1 \dv\Sigma_1
\end{bmatrix}
\end{align}
with the $\dv 0$ in its principal diagonal denoting the $n_0{\times}n_0$-all zero matrix.
\end{corollary}
\noindent In particular, for the setting of testing against independence, i.e., $\dv Y_0=\emptyset$ and the detector's task reduced to guessing whether $\dv Y_1$ and $\dv X$ are independent or not, the optimal trade-off expressed by~\eqref{rate-exponent-region-one-encoder-testing-against-conditional-independence-vector-Gaussian-case} reduces to the set of $(R_1,E)$ pairs that satisfy, for some $n_1{\times}n_1$ matrix $\dv\Omega_1$ such that $\dv 0 \preceq \dv\Omega_1 \preceq \dv\Sigma_1^{-1}$,
\begin{equation}
E \leq \min \left\{R_1 + \log \left|\dv I - \dv\Omega_1 \dv\Sigma_1 \right|, \:\: \log \left| \dv I + \dv\Sigma_{\dv x}  \dv H_1^\dagger \dv\Omega_1 \dv H_1 \right| \right\}.
\label{optimal-rate-exponent-region-Gaussian-one-encoder-testing-against-independence}
\end{equation}
Observe that~\eqref{rate-exponent-region-one-encoder-testing-against-conditional-independence-vector-Gaussian-case} is the counter-part, to the vector Gaussian setting, of the result of~\cite[Theorem 3]{RW12} which provides a single-letter formula for the Type II error exponent for the one-encoder DM testing against conditional independence problem. Similarly,~\eqref{optimal-rate-exponent-region-Gaussian-one-encoder-testing-against-independence} is the solution of the vector Gaussian version of the one-encoder DM testing against independence problem which is studied, and solved, by Ahlswede and Csisz\'ar in~\cite[Theorem 2]{AC86}. 

%.......................................................................................
\subsubsection{The $K$-encoder scalar Gaussian HT problem against independence}~\label{secIII_subsecB_subsubsection2}
%.......................................................................................
Consider now the special case of the setup of Theorem~\ref{theorem-rate-exponent-region-gaussian-hypothesis-testing-against-conditional-independence} in which $K \geq 2$, $Y_0 = \emptyset$, and the sources and noises are all scalar complex-valued Gaussian, i.e., $n_x=1$ and $n_k=1$ for all $k \in \mc K$. The vector $(Y_1,\hdots,Y_K)$ and $X$ are correlated under the null hypothesis $H_0$ with
\begin{equation}
H_0 \: : Y_k = X + Z_k, \quad\text{for all}\:\: k \in \mc K\\
\label{distributions-under-null-hypothesis-ht-against-independence-scalar-Gaussian}
\end{equation}
The noises $Z_1,\hdots,Z_K$ are zero-mean jointly Gaussian, mutually independent and independent from $X$. Also, we assume that the variances $\sigma^2_k$ of noise $Z_k$, $k \in \mc K$, and $\sigma^2_X$ of $X$ are all non-negative. Under the alternate hypothesis $H_1$, the joint distribution of $(X,Y_1,\hdots,Y_K)$, denoted as $Q_{X,Y_1,\hdots,Y_K}$, factorizes as
\begin{equation}
H_1 \: :  Q_{X, Y_1,\hdots,Y_K} = Q_{X} Q_{Y_1,\hdots,Y_K}
\label{distributions-under-alternate-hypothesis-ht-against-conditional-independence-vector-Gaussian}
\end{equation}
where $Q_{X}=P_{X}$, i.e., complex Gaussian with zero-mean and variance $\sigma^2_X$ and $Q_{Y_1,\hdots,Y_K}=P_{Y_1,\hdots,Y_K}$ where $P_{Y_1,\hdots,Y_K}$ is the joint distribution of the vector $(Y_1,\hdots,Y_K)$ under $H_0$ as induced by~\eqref{distributions-under-null-hypothesis-ht-against-independence-scalar-Gaussian}.

\noindent In this case, the result of Theorem~\ref{theorem-rate-exponent-region-gaussian-hypothesis-testing-against-conditional-independence} reduces as stated in the following corollary.

\begin{corollary}~\label{corollary-rates-exponent-region-K-encoder-scalar-Gaussian-HT-against-independence}
For the $K$-encoder scalar Gaussian HT against independence problem described by~\eqref{distributions-under-null-hypothesis-ht-against-independence-scalar-Gaussian} and~\eqref{distributions-under-alternate-hypothesis-ht-against-conditional-independence-vector-Gaussian}, the rate-exponent region is given by the set of all non-negative tuples $(R_1,\hdots,R_K, E)$ that satisfy
\begin{align}
 \mc R_{\text{SG-HT}} = \Big\{  (R_1,\hdots,R_K,E)\::\: &\exists \: (\gamma_1,\hdots,\gamma_K) \in \mathbb{R}^K_{+} \:\: \text{s.t.} \nonumber\\
&\gamma_k \leq \frac{1}{\sigma^2_k},\:\forall k\in \mc K,\:\: \text{and} \:\: \forall \: \mc S \subseteq \mc K \nonumber\\
&\sum_{k \in \mc S} R_k \geq E + \log\Big[\Big(\Big(1+\sigma^2_X\sum_{k \in \mc S^c}\gamma_k\Big)\prod_{k\in \mc S}(1-\gamma_k \sigma^2_k)\Big)^{-1}\Big]\Big\}.
\label{rate-exponent-region-ht-against-independence-scalar-Gaussian}
\end{align}
\end{corollary}

\noindent The region $\mc R_{\text{SG-HT}}$ as given by~\eqref{rate-exponent-region-ht-against-independence-scalar-Gaussian} can be used to, e.g., characterize the centralized rate region, i.e., the set of rate vectors $(R_1,\hdots,R_K)$ that achieve the centralized Type II error exponent 
\begin{equation}
I(Y_1,\hdots,Y_K;X) = \sum_{k=1}^K \log \frac{\sigma^2_X}{\sigma^2_k}.
\end{equation}

 We close this section by mentioning that, as it can be seen from the proof of Theorem~\ref{theorem-rate-exponent-region-gaussian-hypothesis-testing-against-conditional-independence}, the Quantize-Bin-Test scheme of~\cite[Theorem 1]{RW12} evaluated with Gaussian test channels and no time-sharing is optimal for the vector Gaussian $K$-encoder hypothesis testing against conditional independence problem described by~\eqref{distributions-under-null-hypothesis-ht-against-conditional-independence-vector-Gaussian} and~\eqref{distributions-under-alternate-hypothesis-ht-against-conditional-independence-vector-Gaussian}. Furthermore, we note that Rahman and Wagner also characterized the optimal rate-exponent region of a different\footnote{This problem is related to the Gaussian many-help-one problem~\cite{O05,PTR04,WTV08}. Here, different from the setup of Figure~\ref{fig-distributed-hypothesis-testing}, the source $X$ is observed directly by a \textit{main encoder} who communicates with a detector that observes $Y$ in the aim of making a decision on whether $X$ and $Y$ are independent or not. Also, there are helpers that observe independent noisy versions of $X$ and communicate with the detector in the aim of facilitating that test.} Gaussian hypothesis testing against independence problem, called the Gaussian many-help-one hypothesis testing against independence problem therein, in the case of scalar valued sources~\cite[Theorem 7]{RW12}. Specialized to the case $K=1$, the result of Theorem~\ref{theorem-rate-exponent-region-gaussian-hypothesis-testing-against-conditional-independence} recovers that of~\cite[Theorem 7]{RW12} in the case of no helpers; and extends it to vector-valued sources and testing against conditional independence in that case.

\vspace{-0.2cm}

%=======================================================================================
\section{Testing Under Gaussian Noise: Dual Roles of Power and Entropy Power}~\label{secIV}
%=======================================================================================

\vspace{-0.2cm}

In this section, we broaden our view to study a generalization of the $K$-encoder scalar Gaussian HT against independence problem described by~\eqref{distributions-under-null-hypothesis-ht-against-independence-scalar-Gaussian} and~\eqref{distributions-under-alternate-hypothesis-ht-against-conditional-independence-vector-Gaussian} in which the sensors' observations $(Y_1,\hdots,Y_K)$ under the null hypothesis are still independent noisy versions of $X$, with Gaussian noises, but $X$ itself is an arbitrary continuous memoryless source. In particular, $X$ is not necessarily Gaussian. Throughout this section we assume that $X$ has density $P_X(x)$ (not necessarily Gaussian) which has finite differential entropy $h(X)$, variance $\sigma^2_X$ and non-zero finite entropy power
\begin{equation}
N(X) = \frac{e^{2h(X)}}{2{\pi}e}.
\label{definition-entropy-power}
\end{equation}
\noindent Specifically,  $X$ and $(Y_1,\hdots,Y_K)$ are correlated under the null hypothesis $H_0$  with
\begin{equation}
H_0 \: : Y_k = X + Z_k, \quad\text{for}\:\: k=1,\hdots,K
\label{distribution-under-null-hypothesis-ht-against-independence-2}
\end{equation}
where the noise $Z_k$ is zero-mean Gaussian with variance $\sigma_k^2$ and is independent from all other noises and from $X$; and they are independent under $H_1$ with their joint distribution given by
\begin{equation} 
H_1 \: :  Q_{X,Y_1,\hdots,Y_K} = P_X P_{Y_1,\hdots,Y_K}
\label{distribution-under-alternate-hypothesis-ht-against-independence-2}
\end{equation}
where $P_X$ is the distribution of $X$ under $H_0$ (not necessarily Gaussian!) and $P_{Y_1,\hdots,Y_K}$ is the joint distribution of $(Y_1,\hdots,Y_K)$ under $H_0$ as induced by~\eqref{distribution-under-null-hypothesis-ht-against-independence-2}.

\noindent In this section sometimes we will be interested in the sum-rate exponent function, which is defined as
\begin{equation}
R_{\text{sum}}(E) = \min_{(R_1,\hdots,R_K,E)\: \in \:  \mc R_{\text{HT},\epsilon}} \sum_{k=1}^K R_k.
\end{equation}

\noindent Throughout it will be convenient to use the following shorthand notation. For any non-empty subset $\mc S \subseteq \mc K$ the sufficient statistic for $X$ given $\{Y_k\}_{k \in \mc S}$ is given by
\begin{subequations}
\begin{align}
Y(\mc S) &= \frac{1}{|\mc S|} \sum_{k \in \mc S} \frac{\sigma^2_{\mc S}}{\sigma^2_k} Y_k \\
&= X + Z(\mc S)
\end{align}
\label{definition-sufficient-statistic}
\end{subequations}
where
\begin{equation}
Z(\mc S) = \frac{1}{|\mc S|} \sum_{k \in \mc S} \frac{\sigma^2_{\mc S}}{\sigma^2_k} Z_k
\label{harmonic-mean-noises}
\end{equation}
is a zero-mean Gaussian random variable of variance $\sigma^2_{\mc S}/|\mc S|$, and $\sigma^2_{\mc S}$ denotes the harmonic mean of the noise variances in the set $\mc S$, given by
\begin{equation}
\sigma^2_{\mc S} = \left(\frac{1}{|\mc S|} \sum_{k \in \mc S} \frac{1}{\sigma^2_k}\right)^{-1}.
\label{harmonic-mean-variances}
\end{equation}
For the special case of empty set $\mc S =\emptyset$, we set $Y(\mc \emptyset)=Z(\emptyset)=\text{constant}$.

In the rest of this section, we will develop bounds on the rate-exponent region of this model which exhibit a pleasant duality between power and entropy power. Bounds of the same kind of duality were already observed in the context of source coding under the classic squared error distortion measure for point-to-point~\cite[p. 338]{CT06} and multiterminal CEO~\cite{EG19} settings. The recent work~\cite{SV21} is somewhat related, but to a lesser extent. 

\vspace{-0.2cm}

%---------------------------------------------------------------------------------------
\subsection{Special Case $K=1$}~\label{secIV_subsecA}
%---------------------------------------------------------------------------------------
Set $K=1$ in~\eqref{distribution-under-null-hypothesis-ht-against-independence-2} and~\eqref{distribution-under-alternate-hypothesis-ht-against-independence-2}. For notational convenience, we use the substitutions $Y=Y_1$, $Z=Z_1$ and $\sigma^2_Z=\sigma^2_1$. First let us recall that for a given non-negative rate $R$ the optimal rate exponent is given by~\cite[Theorem 2]{AC86}
\vspace{-0.3cm}
\begin{equation}
E(R) = \max_{P_{U|Y} \::\: I(U;Y) \: \leq R} \:\: I(U;X).
\label{optimal-exponent-rate-function-point-to-point-model}
\end{equation}
It is rather easy to see that a simple lower bound on the exponent-rate function is given by
\begin{equation}
E(R) \geq \frac{1}{2} \log^+ \left(\frac{\sigma^2_Y}{\sigma^2_Xe^{-2R}+ \sigma^2_Z}\right).
\label{lower-bound-EPI-exponent-rate-function-point-to-point-model}
\end{equation}
This can be obtained by evaluating the right hand side (RHS) of~\eqref{optimal-exponent-rate-function-point-to-point-model} using the choice of auxiliary
\begin{equation}
U = Y + V
\end{equation}
where V is zero-mean Gaussian with variance
\begin{equation}
\sigma^2_V = \frac{\sigma^2_X+\sigma^2_Z}{e^{2R}-1}
\end{equation}
and is independent from $(X,Z)$.

\noindent Also, it can be shown (see Appendix~\ref{appendix-proof-upper-bound-EPI-exponent-rate-function-point-to-point-model}) that 
\begin{equation}
E(R) \leq \frac{1}{2} \log^+ \left(\frac{N(Y)}{N(X)e^{-2R}+ \sigma^2_Z}\right).
\label{upper-bound-EPI-exponent-rate-function-point-to-point-model}
\end{equation}
Part of the appeal of these bounds is the interesting duality that is played by the source power and its entropy power. Also, this directly implies their tightness in the special case in which the source $X$ is Gaussian since power and entropy power are equal in that case. Moreover, the above also implies that among all sources with the same variance (power) the Gaussian is the worst (i.e., has the smallest 
Type-II error exponent for given $R$). Conversely, among all sources with the same entropy power the Gaussian is the best (i.e., has the largest Type-II error exponent for given $R$). 

% \noindent We close this section by mentioning that bounds of this same kind of duality were already observed in the context of source coding under the classic squared error distortion measure for point-to-point~\cite[p. 338]{CT06} and multiterminal~\cite{EG19} settings.

\vspace{-0.2cm}

%---------------------------------------------------------------------------------------
\subsection{Upper Bound}~\label{secIV_subsecB}
%---------------------------------------------------------------------------------------

\vspace{-0.4cm}

\noindent We now turn to the $K$-encoder test described by~\eqref{distribution-under-null-hypothesis-ht-against-independence-2} and~\eqref{distribution-under-alternate-hypothesis-ht-against-independence-2}. The main result of this section is an upper bound on the exponent-rate function  for an arbitrary continuous source $X$ with finite differential entropy. Its strength is in that a direct consequence of it (Corollary~\ref{corollary-lower-bound-sum-rate-exponent-function-distributed-setting} below) is shown to reflect the right behavior as a function of the number of observations/sensors.

\noindent Recall the definition of the sufficient statistic $Y(\mc S)$ for $X$ given $Y_{\mc S}=\{Y_k\}_{k \in \mc S}$ as given by~\eqref{definition-sufficient-statistic} for given $\mc S \subseteq \mc K$.

\begin{theorem}~\label{theorem-upper-bound-rate-exponent-function-distributed-setting}
If a rate-exponent tuple $(R_1,\ldots,R_K,E)$ is achievable, i.e., $(R_1,\ldots,R_K,E) \in \mc R_{\text{HT}}$, then there must exist non-negative real numbers $(\gamma_1,\hdots,\gamma_K)$ with $\gamma_k \leq 1/\sigma^2_k$ for all $k \in \mc K$ such that for all (strict) subsets $\mc S \subset \mc K$, we have
\begin{align}
E &\leq \frac{1}{2} \log\left(|\mc S^c|\frac{N(Y(\mc S^c))}{\sigma^2_{\mc S^c}} - N(X)\sum_{k \in \mc S^c}   \left(\frac{1}{\sigma^2_k} - \gamma_k\right)\right) \nonumber\\
& \qquad \qquad + \sum_{k \in \mc S} \left(R_k - \frac{1}{2}\log \frac{1}{1-\gamma_k \sigma^2_k}\right);
\label{upper-bound-rate-exponent-function-distributed-setting-strict-subset}
\end{align}
and for the full set $\mc S = \mc K$ we have 
\begin{equation}
E \leq \sum_{k=1}^K \left(R_k - \frac{1}{2}\log \frac{1}{1-\gamma_k \sigma^2_k}\right),
\label{upper-bound-rate-exponent-function-distributed-setting-full-set}
\end{equation}
where $Y(\mc S^c)$ and $\sigma^2_{\mc S^c}$ are defined using~\eqref{definition-sufficient-statistic} and~\eqref{harmonic-mean-variances} respectively.
\end{theorem}

\begin{proof}
The proof of Theorem~\ref{theorem-upper-bound-rate-exponent-function-distributed-setting} appears in Appendix~\ref{appendix-proof-theorem-upper-bound-rate-exponent-function-distributed-setting}.  
\end{proof}

\begin{remark}
A simple entropy power inequality argument can be used to show that the term inside the logarithm in the RHS of~\eqref{upper-bound-rate-exponent-function-distributed-setting-strict-subset} is guaranteed to be larger than $1$ for all non-negative choices of $(\gamma_1,\hdots,\gamma_K)$ that satisfy $ 0 \leq \gamma_k \leq 1/\sigma^2_k$ for all $k \in \mc K$, making the expression well defined. 
\end{remark}

We now state the next corollary which provides a lower bound on the exponent-rate function for an arbitrary continuous source $X$ with finite differential entropy, and whose proof, omitted here for brevity, can be obtained easily from that of the direct part of Theorem~\ref{theorem-rate-exponent-region-gaussian-hypothesis-testing-against-conditional-independence} (for instance, see Eq.~\eqref{rate-exponent-region-ht-against-independence-scalar-Gaussian} of Corollary~\ref{corollary-rates-exponent-region-K-encoder-scalar-Gaussian-HT-against-independence}). In fact, while the result of Theorem~\ref{theorem-rate-exponent-region-gaussian-hypothesis-testing-against-conditional-independence} pertains to the case of jointly Gaussian $(X,Y_1,\hdots,Y_K)$, the key argument of its direct part is the Markov lemma~\cite{CT06} whose proof only uses the fact that conditioned on the source sequence $X^n$ the noisy observations $\{Y_k\}_{k=1}^K$ and the auxiliaries $\{U_k\}_{k=1}^K$ are Gaussian. Clearly, this still holds here even though $X$ is not necessarily Gaussian.

\begin{corollary}~\label{corollary-lower-bound-rate-exponent-function-distributed-setting}
If there exist non-negative real numbers $(\gamma_1,\hdots,\gamma_K)$ with $\gamma_k \leq 1/\sigma^2_k$ for all $k \in \mc K$ such that for all (strict) subsets $\mc S \subset \mc K$, we have
\begin{align}
E &\geq \frac{1}{2} \log\left(|\mc S^c|\frac{\sigma^2_{Y(\mc S^c)}}{\sigma^2_{\mc S^c}} - \sigma^2_X\sum_{k \in \mc S^c}   \left(\frac{1}{\sigma^2_k} - \gamma_k\right)\right) \nonumber\\
& \qquad \qquad + \sum_{k \in \mc S} \left(R_k - \frac{1}{2}\log \frac{1}{1-\gamma_k \sigma^2_k}\right)
\label{lower-bound-rate-exponent-function-distributed-setting}
\end{align}
and for the full set $\mc S = \mc K$ we have 
\begin{equation}
E \geq \sum_{k=1}^K \left(R_k - \frac{1}{2}\log \frac{1}{1-\gamma_k \sigma^2_k}\right),
\end{equation}
then the tuple $(R_1,\ldots,R_K,E)$ is achievable, i.e., $(R_1,\ldots,R_K,E) \in \mc R_{\text{HT}}$, where $Y(\mc S^c)$ and $\sigma^2_{\mc S^c}$ are defined using~\eqref{definition-sufficient-statistic} and~\eqref{harmonic-mean-variances} respectively.
\end{corollary}

\noindent Investigating the above bounds of Theorem~\ref{theorem-upper-bound-rate-exponent-function-distributed-setting} and Corollary~\ref{corollary-lower-bound-rate-exponent-function-distributed-setting}, it is interesting to observe a pleasant duality, in the sense that the power (variance) terms of the lower bound are replaced by entropy power terms (note that $\sigma^2_{Y(\mc S^c)}=\sigma^2_X+\sigma^2_{Z(\mc S^c)}$ but we prefer to use the form given in the theorem so as to emphasize such duality). Among other aspects, this directly implies tightness of the bounds in the special case in which the source $X$ is Gaussian; thus providing an alternative proof of Theorem~\ref{theorem-rate-exponent-region-gaussian-hypothesis-testing-against-conditional-independence} in the special case of testing against independence and scalar Gaussian sources. Furthermore, similar to the single-sensor setting of Section~\ref{secIV_subsecA}, the bounds also imply that for given entropy power the Gaussian is the best distribution and for given power the Gaussian is the worst distribution.

\vspace{-0.2cm}

\newpage

%---------------------------------------------------------------------------------------
\subsection{Sum-Rate Exponent Function}~\label{secIV_subsecC}
%---------------------------------------------------------------------------------------

\vspace{-0.2cm}

\noindent For simplicity, we set all the noise variances to be equal, i.e., $\sigma^2_k=\sigma^2_Z$ for all $k \in \mc K$. Note that in this case, the harmonic mean of the noise variances as defined by~\eqref{harmonic-mean-variances} is $\sigma^2_{\mc S}=\sigma^2_Z$ for all $\mc S \subseteq \mc K$. Using Theorem~\ref{theorem-upper-bound-rate-exponent-function-distributed-setting}, we have 
\begin{align}
& \sum_{k=1}^K R_k \: \stackrel{(a)}{\geq} E + \frac{1}{2} \log \prod_{k=1}^K \frac{1}{1-\gamma_k \sigma^2_Z} \\
&\qquad \stackrel{(b)}{\geq} E - \frac{K}{2} \log \frac{1}{K} \sum_{k=1}^K \left(1-\gamma_k \sigma^2_Z\right) \\
&\qquad \stackrel{(c)}{\geq} E + \frac{K}{2} \log \frac{K N(X)}{\sigma^2_Z}\left(\frac{K N(Y(\mc K))}{\sigma^2_{\mc K}} -e^{2E} \right)^{-1}
\label{lower-bound-sum-rate-exponent-function}
\end{align}
where: $(a)$ follows by~\eqref{upper-bound-rate-exponent-function-distributed-setting-full-set}, $(b)$ follows by using Jensen's inequality, and $(c)$ follows by applying~\eqref{upper-bound-rate-exponent-function-distributed-setting-strict-subset} for $\mc S=\emptyset$. 
 
\noindent The result of the next corollary follows directly from~\eqref{lower-bound-sum-rate-exponent-function}.

\begin{corollary}~\label{corollary-lower-bound-sum-rate-exponent-function-distributed-setting}
If a sum-rate exponent pair $(R_{\text{sum}},E)$ is achievable, i.e., $(R_1,\ldots,R_K,E) \in \mc R_{\text{HT}}$ with $(R_1+\hdots+R_K)=R_{\text{sum}}$, then the following holds,
\begin{equation}
R_{\text{sum}} \geq E + \frac{K}{2} \log^{+} \left(\frac{KN(X)}{KN(Y(\mc K)) -\sigma^2_Z e^{2E}}\right)
\label{lower-bound-sum-rate-exponent-function-distributed-setting}
\end{equation}
for $E$ for which $\sigma^2_Z e^{2E} < KN(Y(\mc K))$, where $Y(\mc K)$ is defined using~\eqref{definition-sufficient-statistic}.
\end{corollary}

\noindent Using Corollary~\ref{corollary-lower-bound-rate-exponent-function-distributed-setting}, it is easy to see that for given exponent $E$ for which $\sigma^2_Z e^{2E} < K\sigma^2_{Y(\mc K)}$ the sum-rate exponent is upper bounded as
\begin{equation}
R_{\text{sum}} \leq E + \frac{K}{2} \log \left(\frac{K\sigma^2_X}{K\sigma^2_{Y(\mc K)} -\sigma^2_Z e^{2E}}\right).
\label{upper-bound-sum-rate-exponent-function-distributed-setting}
\end{equation}
 
\vspace{-0.2cm}

%---------------------------------------------------------------------------------------
\subsection{Application}~\label{secIV_subsecD}
%---------------------------------------------------------------------------------------

\vspace{-0.2cm}

\noindent Part of the utility of the results of Theorem~\ref{theorem-upper-bound-rate-exponent-function-distributed-setting} and Corollary~\ref{corollary-lower-bound-sum-rate-exponent-function-distributed-setting} is, e.g., for investigating asymptotic exponent/rates and losses incurred by distributed detection as function of the number of observations. 

\noindent The gap between the bounds~\eqref{lower-bound-sum-rate-exponent-function-distributed-setting} and~\eqref{upper-bound-sum-rate-exponent-function-distributed-setting} is upper-bounded by
\begin{align}
\Delta(K) & := \frac{K}{2} \log^{+}\left( \frac{\sigma^2_X}{N(X)}\left(\frac{KN(Y(\mc K)) -\sigma^2_Z e^{2E}}{K\sigma^2_{Y(\mc K)} -\sigma^2_Z e^{2E}}\right) \right).
\label{gap-lower-upper-bounds-sum-rate-exponent-function}
\end{align}
Recalling that $Y(\mc K)=X+(\sigma_Z/\sqrt{K})G$ where $G \sim \mc N(0,1)$ the behavior of $\Delta(K)$ for large $K$ can be obtained easily using de Bruijn identity type bound for entropy power~\cite[Eq. (15)]{C18}
\begin{equation}
N(Y(\mc K)) \leq N(X) + \frac{\sigma^2_Z}{K} \left(\frac{d}{dt}N(X+\sqrt{t}G)|_{t=0}\right).
\end{equation}

\noindent More precisely, we obtain
\vspace{-0.2cm}
\begin{subequations}
\begin{align}
\label{upper-bound-on-limit-of-gap-between-lower-and-upper-bounds-ineq1}
0 \leq \lim_{K \to \infty} \Delta(K)  &\leq \frac{\sigma^2_Z}{2} \left[ \left(\frac{\kappa_X}{N(X)}-\frac{1}{\sigma^2_X}\right) - e^{2E}\left(\frac{1}{N(X)}-\frac{1}{\sigma^2_X}\right)\right]^{+} \\
& \leq \frac{\sigma^2_Z}{2} \left(\frac{\kappa_{X}}{N(X)}-\frac{1}{\sigma^2_X}\right) 
\label{upper-bound-on-limit-of-gap-between-lower-and-upper-bounds-ineq2}
\end{align}
\end{subequations}
where the scalar coefficient $\kappa_X$ is defined as
\begin{equation}
\kappa_X := \frac{d}{dt} N(X+\sqrt{t}G)|_{t=0}.
\end{equation}
(Note that if $X$ itself is Gaussian, $\kappa_X=1$ and $N(X)=\sigma^2_X$). Figure~\ref{figs-pdf-and-gap-bounds-inverse-Gaussian} depicts the evolution of the RHS of~\eqref{gap-lower-upper-bounds-sum-rate-exponent-function} as a function of the exponent $E$ for an example non-Gaussian distribution, the Wald distribution given by
\begin{equation}
p_X(x) = \left(\frac{\lambda}{2\pi x^3}\right)^{1/2} \exp \frac{-\lambda(x-\mu)^2}{2\mu^2 x}, \quad x \in ]0,+\infty[
\label{definition-Wald-distribution}
\end{equation}
with the scale parameter $\mu$ set to $1$ and the shape parameter $\lambda$ set to $10$. Also shown for comparison, the upper bound on the limit of $\Delta(K)$ at large $K$ as given by~\eqref{upper-bound-on-limit-of-gap-between-lower-and-upper-bounds-ineq2}. Observe that $\Delta(K)$, and so the gap between our bounds~\eqref{lower-bound-sum-rate-exponent-function-distributed-setting} and~\eqref{upper-bound-sum-rate-exponent-function-distributed-setting}, are relatively small for this example. Also, the gap is larger for larger values of $K$ and smaller values of the exponent (intuitively, this is because the Gaussian part of $Y(K)$, which is $(\sum_{i=1}^K Z_i)/K$, is weaker for increasing values of $K$).

 \begin{figure*}[!htpb]
	\begin{subfigure}[b]{0.49\linewidth}
		\centering
		\includegraphics[width=1\linewidth]{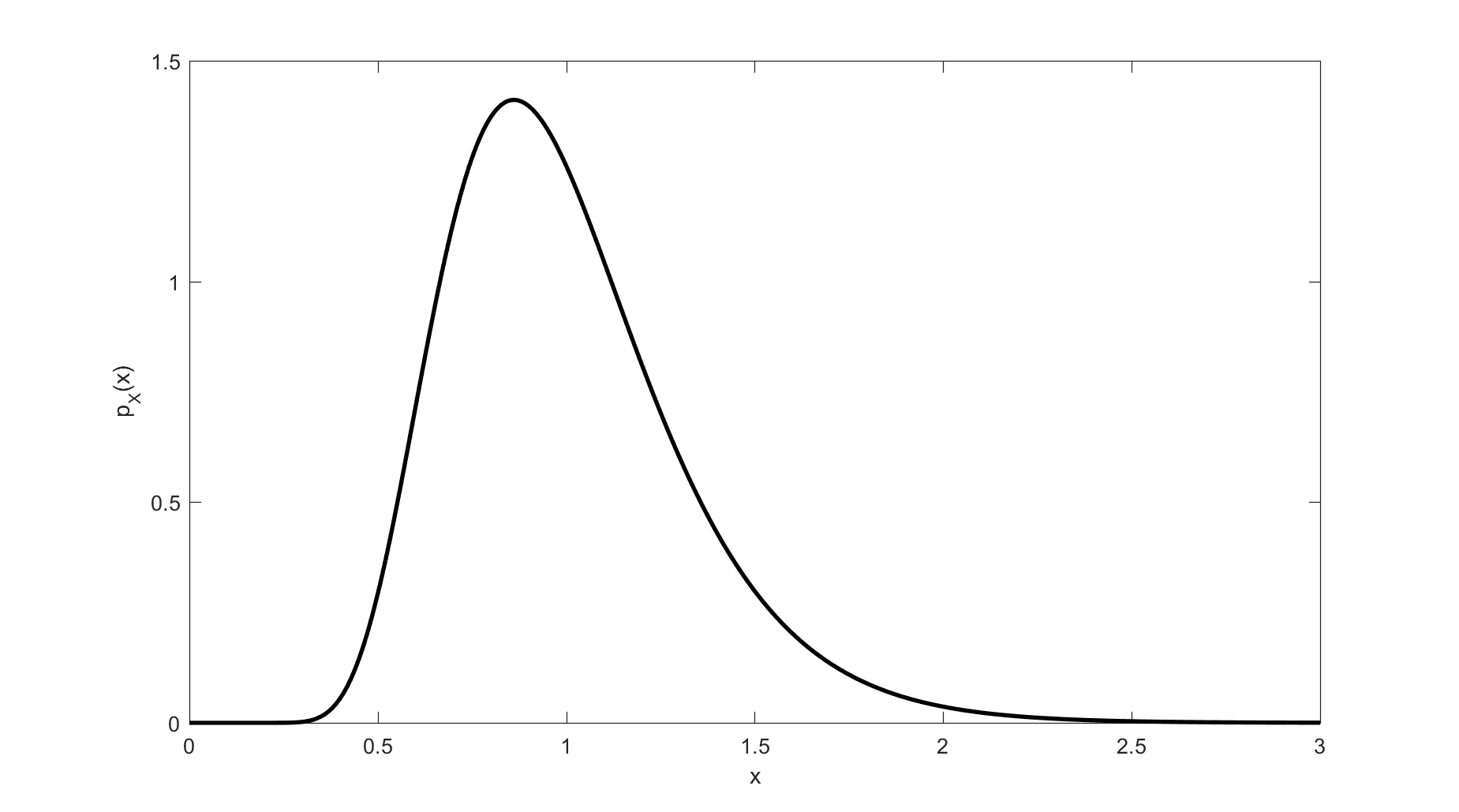}
		\caption{Wald distribution~\eqref{definition-Wald-distribution} with $\mu=1$ and $\lambda=10$}
		\label{fig-pdf-inverse-Gaussian}
	\end{subfigure}
	\begin{subfigure}[b]{0.49\linewidth}
		\centering
		\includegraphics[width=1\linewidth]{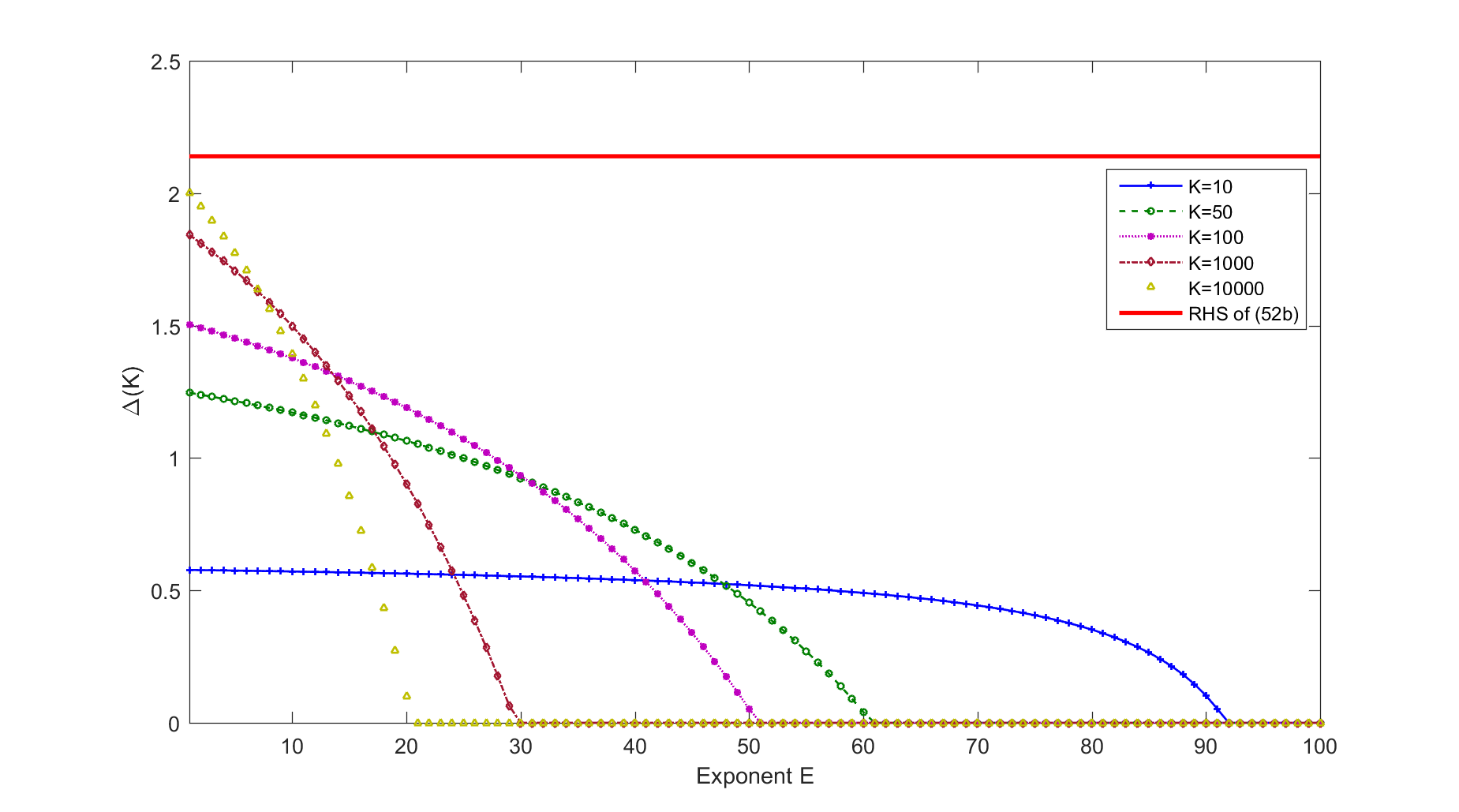}
		\caption{Evolution of $\Delta(K)$ v.s. the exponent }
		\label{fig-gap-bounds-inverse-Gaussian}
	\end{subfigure}
	\caption{Evolution of $\Delta(K)$ as given by Eq. (50) as a function of the exponent $E$ for the Wald distribution with scale parameter $\mu=1$ and shape parameter $\lambda=10$.}
	\label{figs-pdf-and-gap-bounds-inverse-Gaussian}
\end{figure*}

\noindent Consider now a setup with a single sensor that observes the vector $(Y_1,\hdots,Y_K)$. Given that $Y(\mc K)$ is a sufficient statistic for $X$ given $(Y_1,\hdots,Y_K)$, this is equivalent to a point-to-point detection system with a sensor that has $Y(\mc K)$ and a detector that has $X^n$. Let $R$ denote the rate needed to achieve exponent $E$ for this setting. Using~\eqref{lower-bound-sum-rate-exponent-function-distributed-setting} and~\eqref{lower-bound-EPI-exponent-rate-function-point-to-point-model} we get that the cost of distributed processing (rate redundancy) is lower-bounded as
\begin{align}
 (R_{sum} - R)   &\geq \frac{1}{2} \log^{+} \left(\left(\frac{N(X)}{N(Y(\mc K)) -\sigma^2_{Z(\mc K)} e^{2E}}\right)^K  \frac{\sigma^2_{Y(\mc K)} -\sigma^2_{Z(\mc K)} e^{2E}}{\sigma^2_X} \right) \nonumber\\
& = \frac{1}{2} \log^{+} \left(\left(\frac{N(X)}{N(Y(\mc K)) -\frac{\sigma^2_Z}{K} e^{2E}}\right)^K \left(1+\frac{\sigma^2_Z}{K\sigma^2_X}(1-e^{2E})\right)\right). 
\end{align}

%================================================================================================

\appendix
Throughout this section we denote the set of strongly jointly $\epsilon$-typical sequences \cite[Chapter 14.2]{CT91} with respect to the distribution $P_{X,Y}$ as $\mc T_{\epsilon}^n(P_{X,Y})$.

%================================================================================================
\renewcommand{\theequation}{A-\arabic{equation}}
\setcounter{equation}{0}  % reset counter

\subsection{Proof of Proposition~\ref{proposition-entropy-characterization-rate-exponent-region-distributed-ht-against-conditional-independence}}\label{appendix-proposition-entropy-characterization-rate-exponent-region-distributed-ht-against-conditional-independence}

$i)$ Assume that $(R_1,\hdots,R_K,E) \in \mc{R_{\text{HT}}}$. Fix $\epsilon > 0$ and $\delta > 0$ and let a hypothesis test $({\phi}^{(n)}_1, \hdots, {\phi}^{(n)}_K, {\psi}^{(n)}, \mc A_n)$ with Type-I probability of error $(1-\alpha_n) \in [0,1]$ and Type-II probability of error $\beta_n \in [0,1]$ such that 
\begin{subequations}
\begin{align}
\label{proof-of-direct-part-proposition1-step1-a}
\log \|{\phi}^{(n)}_k\| &\leq n(R_k + \delta), \qquad k=1,\hdots,K \\
\label{proof-of-direct-part-proposition1-step1-b}
\alpha_n &\geq (1-\epsilon) \\
- \frac{1}{n} \log \beta_n &\geq  E-\delta.
\label{proof-of-direct-part-proposition1-step1-c}
\end{align}
\label{proof-of-direct-part-proposition1-step1}
\end{subequations}

\noindent First note that we have
\begin{align}
D\Big(P_{{\phi}^{(n)}_1(Y^n_1),\hdots,{\phi}^{(n)}_K(Y^n_K),X^n,Y^n_0}  & \| Q_{{\phi}^{(n)}_1(Y^n_1),\hdots,{\phi}^{(n)}_K(Y^n_K),X^n,Y^n_0} \Big)  \stackrel{(a)}{=} D\Big(P_{{\phi}^{(n)}_1(Y^n_1),\hdots,{\phi}^{(n)}_K(Y^n_K),Y^n_0}  \| Q_{{\phi}^{(n)}_1(Y^n_1),\hdots,{\phi}^{(n)}_K(Y^n_K),Y^n_0} \Big) \nonumber\\
& + \mathbb{E}_{P_{{\phi}^{(n)}_1(Y^n_1),\hdots,{\phi}^{(n)}_K(Y^n_K),Y^n_0}} \Big[ D\Big(P_{X^n | {\phi}^{(n)}_1(Y^n_1),\hdots,{\phi}^{(n)}_K(Y^n_K),Y^n_0} \| Q_{X^n | {\phi}^{(n)}_1(Y^n_1),\hdots,{\phi}^{(n)}_K(Y^n_K),Y^n_0} \Big) \Big] \\
& \stackrel{(b)}{=} D\Big(P_{{\phi}^{(n)}_1(Y^n_1),\hdots,{\phi}^{(n)}_K(Y^n_K),Y^n_0}  \| Q_{{\phi}^{(n)}_1(Y^n_1),\hdots,{\phi}^{(n)}_K(Y^n_K),Y^n_0} \Big) \nonumber\\
& + \mathbb{E}_{P_{{\phi}^{(n)}_1(Y^n_1),\hdots,{\phi}^{(n)}_K(Y^n_K),Y^n_0}} \Big[ D\Big(P_{X^n | {\phi}^{(n)}_1(Y^n_1),\hdots,{\phi}^{(n)}_K(Y^n_K),Y^n_0} \| Q_{X^n |Y^n_0} \Big) \Big] \\
& \stackrel{(c)}{=} D\Big(P_{{\phi}^{(n)}_1(Y^n_1),\hdots,{\phi}^{(n)}_K(Y^n_K),Y^n_0}  \| Q_{{\phi}^{(n)}_1(Y^n_1),\hdots,{\phi}^{(n)}_K(Y^n_K),Y^n_0} \Big) \nonumber\\
& + I({\phi}^{(n)}_1(Y^n_1),\hdots,{\phi}^{(n)}_K(Y^n_K);X^n|Y^n_0) + \mathbb{E}_{P_{Y^n_0}} \Big[ D\Big( P_{X^n|Y^n_0} \| Q_{X^n|Y^n_0} \Big) \Big] \\
& \stackrel{(d)}{\leq} D\Big(P_{Y^n_1,\hdots,Y^n_K,Y^n_0}  \| Q_{Y^n_1,\hdots,Y^n_K,Y^n_0} \Big) \nonumber\\
& + I({\phi}^{(n)}_1(Y^n_1),\hdots,{\phi}^{(n)}_K(Y^n_K);X^n|Y^n_0) + \mathbb{E}_{P_{Y^n_0}} \Big[ D\Big( P_{X^n|Y^n_0} \| Q_{X^n|Y^n_0} \Big) \Big] \\
& \stackrel{(e)}{=} n D\Big(P_{Y_1,\hdots,Y_K,Y_0}  \| Q_{Y_1,\hdots,Y_K,Y_0} \Big) \nonumber\\
& + I({\phi}^{(n)}_1(Y^n_1),\hdots,{\phi}^{(n)}_K(Y^n_K);X^n|Y^n_0) + n \mathbb{E}_{P_{Y_0}} \Big[ D\Big( P_{X|Y_0} \| Q_{X|Y_0} \Big) \Big] \\
& \stackrel{(f)}{=} I({\phi}^{(n)}_1(Y^n_1),\hdots,{\phi}^{(n)}_K(Y^n_K);X^n|Y^n_0)
\label{equality-divergence-conditional-mutual-information}
\end{align}
where: $(a)$ holds by the chain rule for KL divergence; $(b)$ holds since $X^n$ is independent of $(Y^n_1,\hdots,Y^n_K)$ conditionally given $Y^n_0$ under $H_1$; $(c)$ follows by straightforward algebra; $(d)$ holds by the data processing inequality; $(e)$ holds since the vector $(X^n,Y^n_0,Y^n_1,\hdots,Y^n_K)$ is i.i.d. under both $H_0$ and $H_1$; and $(f)$ holds since as per the condition~\eqref{condition-on-same-marginals-under-H0-and-H1} the distributions $P$ and $Q$ have same $(X,Y_0)$- and $(Y_0,Y_1,\hdots,Y_K)$-marginals.

\noindent Thus, for any $\epsilon' > 0$, provided that $\epsilon$ is small enough and $n$ is large enough we get
\begin{align}
I({\phi}^{(n)}_1(Y^n_1),\hdots,{\phi}^{(n)}_K(Y^n_K);X^n|Y^n_0) &\stackrel{(a)}{=} D\Big(P_{{\phi}^{(n)}_1(Y^n_1),\hdots,{\phi}^{(n)}_K(Y^n_K),X^n,Y^n_0}  \| Q_{{\phi}^{(n)}_1(Y^n_1),\hdots,{\phi}^{(n)}_K(Y^n_K),X^n,Y^n_0} \Big) \\
&\stackrel{(b)}{\geq} \alpha_n \log \frac{\alpha_n}{\beta_n} + (1-\alpha_n) \log \frac{1-\alpha_n}{1-\beta_n} \\
& \stackrel{(c)}{=} -h(\alpha_n) - \alpha_n \log \beta_n -(1-\alpha_n) \log (1-\beta_n) \\
& \stackrel{(d)}{\geq} (1-\epsilon)n(E-\delta) - h(\alpha_n) -(1-\alpha_n) \log (1-\beta_n) \\
&= n(E - \epsilon')
\label{proof-of-direct-part-proposition1-step2}
\end{align}
where $(a)$ follows from~\eqref{equality-divergence-conditional-mutual-information}; $(b)$ holds by application of the log-sum inequality~\cite[Theorem 2.7.1]{CT06}; in $(c)$, for $u \in (0,1)$, $h(u)$ denotes the entropy of a Bernoulli-$(u)$ random variable, i.e., 
\begin{equation}
h(u) = - u \log u -(1-u) \log (1-u);
\end{equation}
and $(d)$ holds by using~\eqref{proof-of-direct-part-proposition1-step1-b} and~\eqref{proof-of-direct-part-proposition1-step1-c}.

\noindent The inequalities~\eqref{proof-of-direct-part-proposition1-step1-a} and~\eqref{proof-of-direct-part-proposition1-step2} together show that the tuple $(R_1+\delta,\hdots,R_K+\delta,E-\epsilon') \in \mc R^{\star}$; and, hence, $(R_1,\hdots,R_K,E) \in \xbar{\mc R^{\star}}$. Thus, $\mc{R_{\text{HT}}} \subseteq \xbar{\mc R^{\star}}$.

$ii)$ Assume now that $(R_1,\hdots,R_K,E) \in \xbar{\mc R^{\star}}$. For any $\epsilon > 0$ and $\delta > 0$, since $(R_1 + \delta,\hdots,R_K + \delta, E - \delta) \in \mc R^{\star}$ there must exist $p \in \mathbb{N}$ and functions $(f^{(p)}_1,\hdots,f^{(p)}_K)$ such that
\begin{align}
\label{proof-of-converse-part-proposition1-step1}
 \log \|f^{(p)}_k\| &\leq p (R_k+\delta), \qquad k=1,\hdots,K\\
E - \delta &\leq \frac{1}{p} I(f^{(p)}_1(Y^p_{1,1}),\hdots,f^{(p)}_K(Y^p_{K,1}); X^p|Y^p_{0,1}).
\label{proof-of-converse-part-proposition1-step2}
\end{align}
\noindent By application of Stein's lemma to
\begin{equation}
H_0\: :\: \tilde{P}_{f^{(p)}_1(Y^p_1),\hdots,f^{(p)}_K(Y^p_K),X^p,Y^p_0} \qquad \qquad H_1\: :\: \tilde{Q}_{f^{(p)}_1(Y^p_1),\hdots,f^{(p)}_K(Y^p_K),X^p,Y^p_0}
\label{proof-of-converse-part-proposition1-step3}
\end{equation}
where
\begin{subequations}
\begin{align}
\tilde{P}_{f^{(p)}_1(Y^p_1),\hdots,f^{(p)}_K(Y^p_K),X^p,Y^p_0} &= P_{X^p,Y^p_0} \prod_{k=1}^K P_{f^{(p)}_k(Y^p_k)|X^p,Y^p_0} \\
\tilde{Q}_{f^{(p)}_1(Y^p_1),\hdots,f^{(p)}_K(Y^p_K),X^p,Y^p_0} &= Q_{Y^p_0} Q_{X^p|Y^p_0} Q_{f^{(p)}_1(Y^p_1),\hdots,f^{(p)}_K(Y^p_K)|Y^p_0},
\end{align}
\label{proof-of-converse-part-proposition1-step4}
\end{subequations}
we get for every $\epsilon \in [0,1]$
\begin{equation}
\lim_{l \to \infty}\sup -\frac{1}{l} \log \beta_{\dv R}(lp,\epsilon) \geq D\Big(\tilde{P}_{f^{(p)}_1(Y^p_1),\hdots,f^{(p)}_K(Y^p_K),X^p,Y^p_0}  \| \tilde{Q}_{f^{(p)}_1(Y^p_1),\hdots,f^{(p)}_K(Y^p_K),X^p,Y^p_0} \Big)
\label{proof-of-converse-part-proposition1-step5}
\end{equation}
where $\dv R \triangleq (R_1,\hdots,R_K)$ and 
\begin{align}
\beta_{\dv R}(lp,\epsilon) \triangleq & \min_{(g_1,\hdots, g_K)\:\::\: \log\|g_1\| \: \leq \: lp (R_1 + \delta),\hdots,\log\|g_K\|\: \leq\: lp (R_K+\delta)}  \min_{\mc A} \Big\{\tilde{Q}_{g_1(Y^{lp}_1),\hdots,g_K(Y^{lp}_K),X^{lp},Y^{lp}_0}(\mc A) \: \text{s.t.}: \: \nonumber\\
& \mc A \subset g_1(\mc Y^{lp}_1) \times \hdots \times g_K(\mc Y^{lp}_K) \times \mc X^{lp} \times \mc Y^{lp}_0, \:\: \tilde{P}_{g_1(Y^{lp}_1),\hdots,g_K(Y^{lp}_K),X^{lp},Y^{lp}_0}(\mc A) \geq 1-\epsilon \Big\}.
\label{proof-of-converse-part-proposition1-step6}
\end{align}

\noindent Let, for large $n$, functions ${\phi}^{(n)}_1,\hdots, {\phi}^{(n)}_K$ such that for $1 \leq k \leq K$ the function ${\phi}^{(n)}_k$ is defined over $\mc Y^{n}_{k,1}$ by concatenation from $f^{(p)}_k$ defined over $\mc Y^{p}_{k,1}$ as
\begin{equation}
{\phi}^{(n)}_k(y_{k,1},\hdots,y_{k,n}) \triangleq \Big(f^{(p)}_k(y_{k,1},\hdots, y_{k,p}), \hdots, f^{(p)}_k(y_{k,(l-1)p+1},\hdots, y_{k,lp})\Big), \quad lp \leq n \leq (l+1)p.
\label{proof-of-converse-part-proposition1-step7}
\end{equation}

\noindent Using~\eqref{proof-of-converse-part-proposition1-step1} and~\eqref{proof-of-converse-part-proposition1-step7}, it is easy to see that
\begin{equation}
\log \| {\phi}^{(n)}_k \| \leq n(R_k+\delta), \qquad k=1,\hdots,K.
\label{proof-of-converse-part-proposition1-step8}
\end{equation}

\noindent Also, noting that for $lp \leq n \leq (l+1)p$ we have
\begin{equation}
\beta_{\dv R}((l+1)p,\epsilon) \leq \beta_{\dv R}(n,\epsilon) \leq \beta_{\dv R}(lp,\epsilon)
\label{proof-of-converse-part-proposition1-step9}
\end{equation}
and using~\eqref{proof-of-converse-part-proposition1-step5} it follows that
\begin{align}
\lim_{n \to \infty}\sup -\frac{1}{n} \log \beta_{\dv R}(n,\epsilon) & \geq D\Big(\tilde{P}_{f^{(p)}_1(Y^p_1),\hdots,f^{(p)}_K(Y^p_K),X^p,Y^p_0}  \| \tilde{Q}_{f^{(p)}_1(Y^p_1),\hdots,f^{(p)}_K(Y^p_K),X^p,Y^p_0} \Big) \\
&\stackrel{(a)}{=} I(f^{(p)}_1(Y^p_1),\hdots,f^{(p)}_K(Y^p_{K,1});X^p|Y^p_{0,1}) \\
&\stackrel{(b)}{\geq} E-\delta
\label{proof-of-converse-part-proposition1-step10}
\end{align}
where $(a)$ follows by noticing that by the condition~\eqref{condition-on-same-marginals-under-H0-and-H1} the joint distributions $\tilde{P}_{f^{(p)}_1(Y^p_1),\hdots,f^{(p)}_K(Y^p_K),X^p,Y^p_0}$ and $\tilde{Q}_{f^{(p)}_1(Y^p_1),\hdots,f^{(p)}_K(Y^p_K),X^p,Y^p_0}$ as defined by~\eqref{proof-of-converse-part-proposition1-step4} have same $(X^p,Y^p_0)$- and $(f^{(p)}_1(Y^p_1),\hdots,f^{(p)}_K(Y^p_K),Y^p_0)$-marginals, i.e., $\tilde{P}_{X^p,Y^p_0}=\tilde{Q}_{X^p,Y^p_0}=P_{X^p,Y^p_0}$ and $\tilde{P}_{f^{(p)}_1(Y^p_1),\hdots,f^{(p)}_K(Y^p_K),Y^p_0} = \tilde{Q}_{f^{(p)}_1(Y^p_1),\hdots,f^{(p)}_K(Y^p_K),Y^p_0}=P_{f^{(p)}_1(Y^p_1),\hdots,f^{(p)}_K(Y^p_K),Y^p_0}$ and then applying the steps leading to~\eqref{equality-divergence-conditional-mutual-information}; and $(b)$ holds by using~\eqref{proof-of-converse-part-proposition1-step2}. 

\noindent Now, for convenience let us denote by $\boldsymbol{{\phi}}^{(n)} = ({\phi}^{(n)}_1,\hdots,{\phi}^{(n)}_K)$ and
\begin{align}
\beta(n,\epsilon, \boldsymbol{{\phi}}^{(n)}) \triangleq &  \min_{\mc A} \Big\{\tilde{Q}_{{\phi}^{(n)}_1(Y^{n}_1),\hdots,{\phi}^{(n)}_K(Y^{n}_K),X^{n},Y^{n}_0}(\mc A) \: \text{s.t.}: \: \nonumber\\
& \mc A \subset {\phi}^{(n)}_1(\mc Y^{n}_1) \times \hdots \times {\phi}^{(n)}_K(\mc Y^{n}_K) \times \mc X^{n} \times \mc Y^{n}_0, \:\: \tilde{P}_{{\phi}^{(n)}_1(Y^{n}_1),\hdots,{\phi}^{(n)}_K(Y^{n}_K),X^{n},Y^{n}_0}(\mc A) \geq 1-\epsilon \Big\}.
\label{proof-of-converse-part-proposition1-step11}
\end{align}

\noindent Noticing that as per ~\eqref{proof-of-converse-part-proposition1-step6} the term $\beta_{\dv R}(n,\epsilon)$ of the LHS of~\eqref{proof-of-converse-part-proposition1-step10} involves a minimization over all functions $(g_1,\hdots,g_K)$ for which $\| g_k\| \leq n(R_k+\delta)$, $k=1,\hdots,K$, and recalling that the functions $({\phi}^{(n)}_1, \hdots, {\phi}^{(n)}_K)$ as defined by~\eqref{proof-of-converse-part-proposition1-step7} satisfy~\eqref{proof-of-converse-part-proposition1-step8}, then by~\eqref{proof-of-converse-part-proposition1-step10} we get
\begin{equation}
\lim_{n \to \infty}\sup -\frac{1}{n} \log \beta(n,\epsilon, \boldsymbol{{\phi}}^{(n)}) \geq E-\delta.
\label{proof-of-converse-part-proposition1-step12}
\end{equation}

\noindent Finally, using~\eqref{proof-of-converse-part-proposition1-step8} and~\eqref{proof-of-converse-part-proposition1-step12} it follows that $(R_1,\hdots,R_K,E) \in  \mc{R_{\text{HT}}}$. Thus, $\xbar{\mc R^{\star}} \subseteq \mc{R_{\text{HT}}}$.

%================================================================================================
\renewcommand{\theequation}{B-\arabic{equation}}
\setcounter{equation}{0}  % reset counter
\subsection{Proof of Converse of Theorem~\ref{theorem-rate-exponent-region-hypothesis-testing-DM-case}}\label{appendix-proof-theorem-rate-exponent-region-hypothesis-testing-DM-case}

Let a non-negative tuple $(R_1,\hdots,R_K,E) \in \mc R_{\text{HT}}$ be given. Since $\mc R_{\text{HT}} = \xbar{\mc R^{\star}}$, then there must exist a series of non-negative tuples $\{(R^{(m)}_1,\hdots,R^{(m)}_K,E^{(m)})\}_{m \in \mathbb{N}}$ such that
\begin{subequations}
\begin{align}
& (R^{(m)}_1,\hdots,R^{(m)}_K,E^{(m)}) \in \mc R^{\star} \:\:\: \text{for all}\:\: m \in \mathbb{N}, \quad \text{and} \\
& \lim_{m \to \infty} (R^{(m)}_1,\hdots,R^{(m)}_K,E^{(m)}) =  (R_1,\hdots,R_K,E).
\end{align}
\label{equivalence-HT-CEO-proof-direct-part-step1}
\end{subequations}
Fix $\delta' > 0$. Then, $\exists \:\: m_0 \in \mathbb{N}$ such that for all $m \geq m_0$, we have
\begin{subequations}
 \begin{align}
R_k &\geq R^{(m)}_k - \delta'  \:\:\: \text{for all}\:\: k \in \mc K, \quad \text{and} \\
E &\leq E^{(m)} + \delta'.
\end{align}
\label{equivalence-HT-CEO-proof-direct-part-step2}
\end{subequations}
\noindent For $m \geq m_0$, there exist a series $\{n_m\}_{m \in \mathbb{N}}$ and functions  $\{{\phi}^{(n_m)}_k\}_{k \in \mc K}$ such that 
\begin{subequations}
 \begin{align}
 R^{(m)}_k  &\geq \frac{1}{n_m} \log|{\phi}^{(n_m)}_k|   \:\: \text{for all}\:\: k \in \mc K, \:\: \text{and} \\
E^{(m)} &\leq \frac{1}{n_m} I(\{{\phi}^{(n_m)}_k(Y^{n_m}_k)\}_{k \in \mc K};X^{n_m}|Y^{n_m}_0).
\end{align}
\label{equivalence-HT-CEO-proof-direct-part-step3}
\end{subequations}
\noindent Combining~\eqref{equivalence-HT-CEO-proof-direct-part-step2} and~\eqref{equivalence-HT-CEO-proof-direct-part-step3} we get that for all $m \geq m_0$,
\begin{subequations}
 \begin{align}
 R_k  &\geq \frac{1}{n_m} \log|{\phi}^{(n_m)}_k(Y^{n_m}_k)| - \delta'   \:\: \text{for all}\:\: k \in \mc K, \:\: \text{and} \\
E  &\leq \frac{1}{n_m} I(\{{\phi}^{(n_m)}_k(Y^{n_m}_k)\}_{k \in \mc K};X^{n_m}|Y^{n_m}_0) + \delta'.
\end{align}
\label{equivalence-HT-CEO-proof-direct-part-step4}
\end{subequations}
\noindent The second inequality of~\eqref{equivalence-HT-CEO-proof-direct-part-step4} implies that
\begin{equation}
H(X^{n_m} | \{{\phi}^{(n_m)}_k(Y^{n_m}_k)\}_{k \in \mc K}, Y^{n_m}_0) \leq n_m (H(X|Y_0)-E) + n_m \delta'.
\label{equivalence-HT-CEO-proof-direct-part-step5}
\end{equation}

\noindent Let $\mc S \subseteq \mc K$ a given subset of $\mc K$ and $J_k := {\phi}_k^{(n_m)}(Y_k^{n_m})$. Also, define, for $i=1,\ldots,n_m$, the following auxiliary random variables
\begin{equation}	
	U_{k,i} := (J_k, Y_k^{i-1}), \quad Q_i := (X^{i-1}, X_{i+1}^{n_m}, Y_0^{i-1}, Y_{0,i+1}^{n_m}).
	\label{proof-converse-definition-of-auxiliary-random-variables}
\end{equation}
Note that, for all $k \in \mc K$, it holds that $U_{k,i} \mkv Y_{k,i} \mkv (X_i, Y_{0,i}) \mkv Y_{\mc K \setminus k,i} \mkv U_{\mc K \setminus k,i}$ is a Markov chain in this order.

\noindent We have 
\begin{align}
n_m  \sum_{k\in \mc S} R_k &\geq \sum_{k \in \mc S} H(J_k) \nonumber\\
& \geq H(J_\mc S) \nonumber\\
& \geq  H(J_\mc S|J_{\mc S^c},Y_0^{n_m}) \nonumber\\
& \geq I(J_{\mc S}; X^{n_m}, Y_\mc S^{n_m}|J_{\mc S^c},Y_0^{n_m}) \nonumber\\   
&= I(J_{\mc S}; X^{n_m}|J_{\mc S^c},Y_0^{n_m}) + I(J_{\mc S}; Y_\mc S^n|X^{n_m},J_{\mc S^c},Y_0^{n_m}) \nonumber\\ 
&= H(X^{n_m}|J_{\mc S^c},Y_0^{n_m}) - H(X^{n_m}|J_{\mc K},Y_0^{n_m}) + I(J_{\mc S}; Y_\mc S^{n_m}|X^{n_m},J_{\mc S^c},Y_0^{n_m}) \nonumber\\  
&\stackrel{(a)}{\geq} H(X^{n_m}|J_{\mc S^c}, Y_0^{n_m}) - H(X^{n_m}|Y^{n_m}_0)  + I(J_{\mc S}; Y_\mc S^{n_m}|X^{n_m},J_{\mc S^c},Y_0^{n_m}) + n_m E - n_m \delta' \nonumber\\  
&= \sum_{i=1}^{n_m} H(X_i|J_{\mc S^c},X^{i-1},Y_0^{n_m}) - H(X^{n_m}|Y^{n_m}_0) + I(J_{\mc S}; Y_\mc S^{n_m}|X^{n_m},J_{\mc S^c},Y_0^{n_m}) + n_m E - n_m \delta' \nonumber\\  
&\stackrel{(b)}{\geq} \sum_{i=1}^{n_m} H(X_i|J_{\mc S^c},X^{i-1},X_{i+1}^{n_m},Y_{\mc S^c}^{i-1},Y_0^{n_m}) - H(X^{n_m}|Y^{n_m}_0) + I(J_{\mc S}; Y_\mc S^{n_m}|X^{n_m},J_{\mc S^c},Y_0^{n_m}) + n_m E - n_m \delta' \nonumber\\
&\stackrel{(c)}{=} \sum_{i=1}^{n_m} H(X_i|U_{\mc S^c,i},Y_{0,i},Q_i) - H(X^{n_m}|Y^{n_m}_0) + I(J_{\mc S}; Y_\mc S^{n_m}|X^{n_m},J_{\mc S^c},Y_0^{n_m}) + n_m E - n_m \delta' \nonumber\\
&\stackrel{(d)}{=} I(J_{\mc S}; Y_\mc S^{n_m}|X^{n_m},J_{\mc S^c},Y_0^{n_m}) - \sum_{i=1}^{n_m} I(U_{\mc S^c,i}, X_i|Y_{0,i},Q_i) + n_m E - n_m \delta' 
\label{lower-bounding-sum-rate-step1}
\end{align}
where $(a)$ follows by using~\eqref{equivalence-HT-CEO-proof-direct-part-step5}; $(b)$ holds since conditioning reduces entropy; and $(c)$ follows by substituting using~\eqref{proof-converse-definition-of-auxiliary-random-variables}; and $(d)$ holds since $(X^{n_m},Y^{n_m}_0)$ is memoryless and $Q_i$ is independent of $(X_i,Y_{0,i})$ for all $i=1,\hdots,n_m$. 

\noindent The term $I(J_{\mc S}; Y_\mc S^{n_m}|X^{n_m},J_{\mc S^c},Y_0^{n_m})$ on the RHS of~\eqref{lower-bounding-sum-rate-step1} can be lower bounded as
\begin{align}
 I(J_{\mc S};  Y_{\mc S}^{n_m}|X^{n_m},J_{\mc S^c},Y_0^{n_m}) &\stackrel{(a)}{\geq} \sum_{k \in \mc S} I(J_k;Y_k^{n_m}|X^{n_m},Y_0^{n_m}) \nonumber\\
& = \sum_{k \in \mc S} \sum_{i=1}^{n_m} I(J_k;Y_{k,i}|Y_k^{i-1},X^{n_m},Y_0^{n_m}) \nonumber\\
&\stackrel{(b)}{=} \sum_{k \in \mc S} \sum_{i=1}^{n_m} I(J_k,Y_k^{i-1};Y_{k,i}|X^{n_m},Y_0^{n_m}) \nonumber\\
&\stackrel{(c)}{=} \sum_{k \in \mc S} \sum_{i=1}^{n_m} I(U_{k,i};Y_{k,i}|X_i,Y_{0,i},Q_i)
 \label{lower-bounding-sum-rate-step2}
\end{align}
where $(a)$ follows due to the Markov chain $J_k \mkv Y_k^{n_m} \mkv (X^{n_m}, Y_0^{n_m}) \mkv Y_{\mc S \setminus k}^{n_m} \mkv J_{\mc S \setminus k}$; $(b)$ follows due to the Markov chain $Y_{k,i} \mkv (X^{n_m},Y_0^{n_m}) \mkv Y_k^{i-1}$; and $(c)$ follows by substituting using~\eqref{proof-converse-definition-of-auxiliary-random-variables}. 

\noindent Then, combining~\eqref{lower-bounding-sum-rate-step1} and~\eqref{lower-bounding-sum-rate-step2}, we get
\begin{equation}
 n_m E  \leq \sum_{i=1}^{n_m} I(U_{\mc S^c,i}, X_i|Y_{0,i},Q_i ) + n_m \sum_{k\in \mc S} R_k  - \sum_{k\in \mc S} \sum_{i=1}^{n_m} I(U_{k,i};Y_{k,i}|X_i,Y_{0,i},Q_i) + n_m \delta'.
\label{upper-bounding-exponent-final-step} 
\end{equation}
\noindent Noticing that $\delta'$ in~\eqref{upper-bounding-exponent-final-step} can be chosen arbitrarily small, a standard time-sharing argument completes the proof of the converse part.

% END OF THE PROOF OF THE DIRECT PART OF THEOREM 1

%================================================================================================
\renewcommand{\theequation}{C-\arabic{equation}}
\setcounter{equation}{0}  % reset counter
\subsection{Proof of Theorem~\ref{theorem-rate-exponent-region-gaussian-hypothesis-testing-against-conditional-independence}}\label{appendix-proof-theorem-rate-exponent-region-gaussian-hypothesis-testing-against-conditional-independence}

First note that a characterization (in terms of auxiliaries) of the rate-exponent region of the memoryless vector Gaussian hypothesis testing against conditional independence problem of Section~\ref{secIV}, obtained by an easy extension of the result of Theorem~\ref{theorem-rate-exponent-region-hypothesis-testing-DM-case} to the continuous alphabet case through standard discretization arguments), is given by the union of all non-negative tuples $(R_1,\hdots,R_K,E)$ that satisfy for all $\mc S \subseteq \mc K$, 
\begin{equation}
E - \sum_{k \in \mc S} R_k \leq I(U_{\mc S^c};\dv X|\dv Y_0,Q) - \sum_{k \in \mc S} I(\dv Y_k;U_k|\dv X,\dv Y_0,Q)
\label{rates-exponent-region-vector-Gaussian-model-in-terms-of-auxiliaries}
\end{equation}
for some joint distribution of the form that factorizes as 
\begin{align}
P_{\dv X, \dv Y_0, \dv Y_{\mc K}, U_{\mc K}, Q}(\dv x, \dv y_0, \dv y_{\mc K},u_{\mc K},q) &= P_Q(q)  P_{\dv X, \dv Y_0}(\dv x, \dv y_0) \nonumber\\
& \times \prod_{k=1}^K P_{\dv Y_k|\dv X, \dv Y_0}(\dv y_k|\dv x,\dv y_0) \: \prod_{k=1}^{K} P_{U_k|\dv Y_k,Q}(u_k|\dv y_k,q).
\label{joint-distribution-rates-exponent-region-vector-Gaussian-model-in-terms-of-auxiliaries}
\end{align}

%.......................................................................................
\subsubsection{Converse part}
%.......................................................................................

Let an achievable tuple $(R_1,\hdots,R_K,E)$ be given. Using the above there must exist auxiliary random variables $(U_1,\hdots,U_K,Q)$ with distribution that factorizes as~\eqref{joint-distribution-rates-exponent-region-vector-Gaussian-model-in-terms-of-auxiliaries} such that~\eqref{rates-exponent-region-vector-Gaussian-model-in-terms-of-auxiliaries} holds for any subset $\mc S \subseteq \mc K$. The converse proof of Theorem~\ref{theorem-rate-exponent-region-gaussian-hypothesis-testing-against-conditional-independence} relies on deriving an upper bound on the RHS of~\eqref{rates-exponent-region-vector-Gaussian-model-in-terms-of-auxiliaries}.  In doing so, we use the technique of~\cite[Theorem 8]{EU14} which relies on the de Bruijn identity and the properties of Fisher information; and extend the argument to account for the time-sharing variable $Q$ and side information $\dv Y_0$.

\noindent For convenience, we first state the following lemma.

\begin{lemma}{\cite{DCT91,EU14}}~\label{lemma-fisher}
Let $(\mathbf{X,Y})$  be a pair of random vectors with pmf $p(\mathbf{x},\mathbf{y})$. We have
\begin{equation*}
\log\left|(\pi e) \dv J^{-1}(\dv X|\dv Y)\right| \leq h(\dv X|\dv Y) \leq \log\left|(\pi e) \mathrm{mmse}(\dv X|\dv Y)\right|
\end{equation*}
where the conditional Fisher information matrix is defined as
\begin{equation*}
\dv J(\dv X|\dv Y) := \mathbb{E} \left[\nabla\log p(\dv X|\dv Y) \nabla\log p(\dv X|\dv Y)^\dagger\right]
\end{equation*}
and the minimum mean squared error (MMSE) matrix is 
\begin{equation*}
\mathrm{mmse}(\dv X|\dv Y) := \mathbb{E} \left[\left(\dv X-\mathbb{E}[\dv X|\dv Y]\right)\left(\dv X-\mathbb{E} [\dv X|\dv Y]\right)^\dagger\right]. \qed 
\end{equation*}
\end{lemma}

\noindent Fix $q\in \mc{Q}$ and $\mc S \subseteq \mc Q$. There must exists a matrix $\dv\Omega_{k,q}$ such that $\dv 0 \preceq \dv\Omega_{k,q} \preceq \dv\Sigma_k^{-1}$ and  
\begin{equation}~\label{equation-outer-1-mmse}
\mathrm{mmse}\left(\dv Y_k|\dv X, U_{k,q}, \dv Y_0, q\right) = \dv\Sigma_k - \dv\Sigma_k \dv\Omega_{k,q} \dv\Sigma_k.
\end{equation}
It is easy to see that such $\dv\Omega_{k,q}$ always exists since 
\begin{equation}
\dv 0 \preceq \mathrm{mmse}\left(\dv Y_k|\dv X,U_{k,q},\dv Y_0,q\right) \preceq \dv\Sigma_{\dv y_k|(\dv x, \dv y_0)} = \dv\Sigma_k.
\label{equation-outer-2-covariance-noise}
\end{equation} 

\noindent Then, we have
\begin{align}
 I(\dv Y_k; U_k|\dv X, \dv Y_0, Q=q) &=  \log\left|(\pi e)\dv\Sigma_k\right| - h(\dv Y_k|\dv X, U_{k,q}, \dv Y_0, Q=q) \nonumber\\
&\stackrel{(a)}{\geq} \log|\dv\Sigma_k| - \log\left|\mathrm{mmse}(\dv Y_k|\dv X, U_{k,q}, \dv Y_0, Q=q)\right| \nonumber\\
&\stackrel{(b)}{=} -\log\left|\dv I- \dv\Omega_{k,q}\dv\Sigma_k\right| \label{equation-Gausss-CEO-first-inequality-q}
\end{align}
where $(a)$ is due to Lemma~\ref{lemma-fisher}; and $(b)$ is due to~\eqref{equation-outer-1-mmse}.

\noindent Now, let the matrix $\dv\Lambda_{\bar{\mc S},q}$ be defined as  
\begin{align}~\label{equation-definition-Tq}
\dv\Lambda_{\bar{\mc S},q} :=
\begin{bmatrix}
\dv 0 & \dv 0 \\
\dv 0 & \mathrm{diag}(\{ \dv\Sigma_k - \dv\Sigma_k \dv\Omega_{k,q} \dv\Sigma_k \}_{k\in\mc S^c})
\end{bmatrix}. 
\end{align}
Then, we have 
\begin{align}
 I(U_{\mc S^c}; \dv X|\dv Y_0,Q=q) & = h(\dv X|\dv Y_0) - h\left(\dv X|U_{S^c,q}, \dv Y_0, Q=q\right) \nonumber\\
 &\: \stackrel{(a)}{\leq} h(\dv X|\dv Y_0) - \log\left|(\pi e) \dv J^{-1}\left(\dv X| \dv U_{S^c,q}, \dv Y_0, q\right)\right| \nonumber\\
&\: \stackrel{(b)}{=} h(\dv X|\dv Y_0) - \log\left| (\pi e) \left( \dv\Sigma_{\dv x}^{-1} + \dv H_{\bar{\mc S}}^\dagger \dv\Sigma_{\dv n_{\bar{\mc S}}}^{-1} \left( \dv I - \dv\Lambda_{\bar{\mc S},q} \dv\Sigma_{\dv n_{\bar{\mc S}}}^{-1} \right) \dv H_{\bar{\mc S}} \right)^{-1} \right| \label{equation-Gausss-CEO-second-inequality-q}
\end{align}
where $(a)$ follows by using Lemma~\ref{lemma-fisher}; and for $(b)$ holds by using the equality 
\begin{equation}~\label{equation-Fisher-equality}
\dv J(\dv X|U_{S^c,q}, \dv Y_0, q) = \dv\Sigma_{\dv x}^{-1} + \dv H_{\bar{\mc S}}^\dagger \dv\Sigma_{\dv n_{\bar{\mc S}}}^{-1} \big( \dv I - \dv\Lambda_{\bar{\mc S},q} \dv\Sigma_{\dv n_{\bar{\mc S}}}^{-1} \big) \dv H_{\bar{\mc S}}. 
\end{equation}
the proof of which uses a connection between MMSE and Fisher information as shown next. More precisely, for the proof of~\eqref{equation-Fisher-equality} first recall de Brujin identity which relates Fisher information and MMSE. 

\begin{lemma}{\cite{EU14}}~\label{lemma-Brujin}
Let $(\dv V_1,\dv V_2)$ be a random vector with finite second moments and $\dv Z\sim\mc{CN}(\dv 0, \dv\Sigma_{\dv z})$ independent of $(\dv V_1,\dv V_2)$. Then
\begin{equation*}
\mathrm{mmse}\left(\dv V_2|\dv V_1,\dv V_2+\dv Z\right) = \dv\Sigma_{\dv z} - \dv\Sigma_{\dv z} \dv J\left(\dv V_2+\dv Z|\dv V_1\right) \dv\Sigma_{\dv z}. \qed 
\end{equation*}
\end{lemma}

\noindent From MMSE estimation of Gaussian random vectors, we have  
\begin{align}
\dv X &= \mathbb{E} [\dv X|\dv Y_{\bar{\mc S}}] + \dv W_{\bar{\mc S}} \nonumber\\
&= \dv G_{\bar{\mc S}} \dv Y_{\bar{\mc S}} + \dv W_{\bar{\mc S}} 
\label{equation-outer-3}
\end{align}
where $\dv G_{\bar{\mc S}} := \dv\Sigma_{\dv w_{\bar{\mc S}}} \dv H_{\bar{\mc S}}^\dagger \dv\Sigma_{\dv n_{\bar{\mc S}}}^{-1}$, and $\dv W_{\bar{\mc S}} \sim \mathcal{CN}(\dv 0, \dv\Sigma_{\dv w_{\bar{\mc S}}})$ is a Gaussian vector that is independent of $\dv Y_{\bar{\mc S}}$ and  
\begin{equation}~\label{equation-covariance-w}
\dv\Sigma_{\dv w_{\bar{\mc S}}}^{-1} :=  \dv\Sigma_{\dv x}^{-1} + \dv H_{\bar{\mc S}}^\dagger \dv\Sigma_{\dv n_{\bar{\mc S}}}^{-1} \dv H_{\bar{\mc S}}. 
\end{equation}

\noindent Next, we show that the cross-terms of $\mathrm{mmse}\left(\dv Y_{\mc S^c}|\dv X,U_{\mc S^{c},q},\dv Y_0,q \right)$ are zero. For $i\in\mc S^c$ and $j\neq i$, we have 
\begin{align}
& \mathbb{E} \left[\left(Y_i-\mathbb{E}[Y_i|\dv X,U_{\mc S^{c},q},\dv Y_0,q]\right) \left(Y_j-\mathbb{E}[Y_j|\dv X,U_{\mc S^{c},q},\dv Y_0,q]\right)^\dagger\right] \nonumber\\
&\qquad \qquad \stackrel{(a)}{=}  \mathbb{E}\left[ \mathbb{E}\left[ \left(Y_i-\mathbb{E}[Y_i|\dv X,U_{\mc S^{c},q},\dv Y_0,q]\right) \left(Y_j-\mathbb{E}[Y_j|\dv X,U_{\mc S^{c},q},\dv Y_0,q]\right)^\dagger|\dv X,\dv Y_0 \right] \right] \nonumber\\
&\qquad \qquad \stackrel{(b)}{=}  \mathbb{E}\left[ \mathbb{E}\left[ \left(Y_i-\mathbb{E}\left[Y_i|\dv X,U_{\mc S^{c},q},\dv Y_0,q\right]\right)|\dv X,\dv Y_0 \right] \times \mathbb{E}\left[\left(Y_j-\mathbb{E}[Y_j|\dv X,U_{\mc S^{c},q},\dv Y_0,q]\right)^\dagger|\dv X,\dv Y_0 \right] \right] \nonumber\\
&\qquad \qquad   = \dv 0,
\label{equation-cross-terms}
\end{align}
where $(a)$ is due to the law of total expectation; $(b)$ is due to the Markov chain $\dv Y_k \mkv (\dv X,\dv Y_0) \mkv \dv Y_{\mc K \setminus k}$. Then, we have
\begin{align}
\mathrm{mmse}\big(\dv G_{\bar{\mc S}} \dv Y_{\bar{\mc S}} \big|\dv X,U_{\mc S^c,q},\dv Y_0,q \big) &= \dv G_{\bar{\mc S}}  \: \mathrm{mmse}\left(\dv Y_{\bar{\mc S}}|\dv X,U_{\mc S^c,q},\dv Y_0,q \right) \dv G_{\bar{\mc S}}^\dagger \nonumber\\
&\stackrel{(a)}{=} 
\dv G_{\bar{\mc S}}
\begin{bmatrix}
\dv 0 & \dv 0 \\
\dv 0 & \mathrm{diag}(\{\mathrm{mmse}(\dv Y_k|\dv X,U_{\mc S^c,q},\dv Y_0,q)\}_{k\in\mc S^c})
\end{bmatrix}
\dv G_{\bar{\mc S}}^\dagger \nonumber\\ 
&\stackrel{(b)}{=} \dv G_{\bar{\mc S}} \dv\Lambda_{\bar{\mc S},q} \dv G_{\bar{\mc S}}^\dagger \label{equation-outer-4}
\end{align}
where $(a)$ follows since the cross-terms are zero as shown in~\eqref{equation-cross-terms}; and $(b)$ follows due to~\eqref{equation-outer-1-mmse} and the definition of $\dv\Lambda_{\bar{\mc S},q}$ given in~\eqref{equation-definition-Tq}.

\vspace{0.2cm}

\noindent We note that $\dv W_{\bar{\mc S}}$ is independent of $\dv Y_{\bar{\mc S}}=(\dv Y_0, \dv Y_{{\mc S}^c})$; and, with the Markov chain $U_{{\mc S}^c} \mkv \dv Y_{{\mc S}^c} \mkv (\dv X, \dv Y_0)$, which itself implies $U_{{\mc S}^c} \mkv \dv Y_{{\mc S}^c} \mkv (\dv X, \dv Y_0,\dv W_{\bar{\mc S}})$, this yields that $\dv W_{\bar{\mc S}}$ is independent of $U_{{\mc S}^c}$. Thus, $\dv W_{\bar{\mc S}}$ is independent of $(\dv G_{\bar{\mc S}} \dv Y_{\bar{\mc S}},U_{{\mc S}^c},\dv Y_0,Q)$. Applying Lemma~\ref{lemma-Brujin} with 
\begin{subequations}
\begin{align}
\dv V_1 &:= (U_{{\mc S}^c},\dv Y_0,Q) \\
\dv V2 &:=\dv G_{\bar{\mc S}} \dv Y_{\bar{\mc S}} \\
\dv Z &:= \dv W_{\bar{\mc S}},
\end{align}
\end{subequations}
 we get   
\begin{align*}    
\dv J(\dv X|U_{S^c,q},\dv Y_0,q) &= \dv\Sigma_{\dv w_{\bar{\mc S}}}^{-1} - \dv\Sigma_{\dv w_{\bar{\mc S}}}^{-1} \: \mathrm{mmse} \big( \dv G_{\bar{\mc S}} \dv Y_{\bar{\mc S}} \big| \dv X,U_{\mc S^{c},q},\dv Y_0,q \big) \dv\Sigma_{\dv w_{\bar{\mc S}}}^{-1} \\
&\stackrel{(a)}{=} \dv\Sigma_{\dv w_{\bar{\mc S}}}^{-1} - \dv\Sigma_{\dv w_{\bar{\mc S}}}^{-1} \dv G_{\bar{\mc S}} \dv\Lambda_{\bar{\mc S},q} \dv G_{\bar{\mc S}}^\dagger \dv\Sigma_{\dv w_{\bar{\mc S}}}^{-1} \\
&\stackrel{(b)}{=} \dv\Sigma_{\dv x}^{-1} + \dv H_{\bar{\mc S}}^\dagger \dv\Sigma_{\dv n_{\bar{\mc S}}}^{-1} \dv H_{\bar{\mc S}} -  \dv H_{\bar{\mc S}}^\dagger \dv\Sigma_{\dv n_{\bar{\mc S}}}^{-1} \dv\Lambda_{\bar{\mc S},q} \dv\Sigma_{\dv n_{\bar{\mc S}}}^{-1} \dv H_{\bar{\mc S}}\\
&= \dv\Sigma_{\dv x}^{-1} + \dv H_{\bar{\mc S}}^\dagger \dv\Sigma_{\dv n_{\bar{\mc S}}}^{-1} \left( \dv I - \dv\Lambda_{\bar{\mc S},q} \dv\Sigma_{\dv n_{\bar{\mc S}}}^{-1}  \right) \dv H_{\bar{\mc S}},
\end{align*}
where $(a)$ is due to~\eqref{equation-outer-4}; and $(b)$ follows due to the definitions of $\dv\Sigma_{\dv w_{\bar{\mc S}}}^{-1}$ and $\dv G_{\bar{\mc S}}$. This proves~\eqref{equation-Fisher-equality}.

\noindent Next, averaging over the time sharing random variable $Q$ both sides of the inequalities~\eqref{equation-Gausss-CEO-first-inequality-q} and~\eqref{equation-Gausss-CEO-second-inequality-q} and letting $\dv\Omega_k := \sum_{q\in \mc Q}p(q) \dv\Omega_{k,q}$ we get
\begin{align}
I(\dv Y_k;\dv U_k|\dv X,\dv Y_0,Q) &= \sum_{q \in \mc Q} p(q) I(\dv Y_k;\dv U_k|\dv X,\dv Y_0,Q=q)  \nonumber\\
&\stackrel{(a)}{\geq} - \sum_{q \in \mc Q} p(q) \log|\dv I- \dv\Omega_{k,q} \dv\Sigma_k| \nonumber\\
&\stackrel{(b)}{\geq} -\log |\dv I - \sum_{q \in \mc Q} p(q) \dv\Omega_{k,q} \dv\Sigma_k| \nonumber\\  
&= -\log |\dv I- \dv\Omega_k \dv\Sigma_k| 
\label{equation-Gausss-CEO-first-inequality}
\end{align}
where $(a)$ follows from~\eqref{equation-Gausss-CEO-first-inequality-q}; and $(b)$ follows from the concavity of the log-det function and Jensen's Inequality. 

\noindent Besides,  we have
\begin{align}
I(U_{\mc S^c}; \dv X|\dv Y_0,Q) &= h(\dv X|\dv Y_0) - \sum_{q \in \mc Q} p(q) h(\dv X|U_{S^c,q}, \dv Y_0, Q=q) \nonumber\\ 
&\stackrel{(a)}{\leq} h(\dv X|\dv Y_0) - \sum_{q \in \mc Q} p(q) \log\left| (\pi e) \left(\dv\Sigma_{\dv x}^{-1} + \dv H_{\bar{\mc S}}^\dagger \dv\Sigma_{\dv n_{\bar{\mc S}}}^{-1} \big( \dv I - \dv\Lambda_{\bar{\mc S},q} \dv\Sigma_{\dv n_{\bar{\mc S}}}^{-1}  \big) \dv H_{\bar{\mc S}} \right)^{-1} \right| \nonumber\\
&\stackrel{(b)}{\leq} h(\dv X|\dv Y_0) - \log\left| (\pi e) \left( \dv\Sigma_{\dv x}^{-1} + \dv H_{\bar{\mc S}}^\dagger \dv\Sigma_{\dv n_{\bar{\mc S}}}^{-1} \big( \dv I - \dv\Lambda_{\bar{\mc S}} \dv\Sigma_{\dv n_{\bar{\mc S}}}^{-1}  \big) \dv H_{\bar{\mc S}} \right)^{-1} \right|, 
\label{equation-Gausss-CEO-second-inequality}  
\end{align}
where $(a)$ is due to~\eqref{equation-Gausss-CEO-second-inequality-q}; and $(b)$ is due to the concavity of the log-det function and Jensen's inequality and the definition of $\dv\Lambda_{\bar{\mc S}}$ given in~\eqref{equation-definition-T}.

\noindent Using~\eqref{equation-Gausss-CEO-second-inequality}, we get
\begin{align}
I(U_{\mc S^c}; \dv X|\dv Y_0,Q) &= \mathbb{E}_Q \Big[ I\left(U_{\mc S^c}; \dv X|\dv Y_0,Q=q\right) \Big] \nonumber\\ 
&\leq h(\dv X|\dv Y_0) - \log\left| (\pi e) \left( \dv\Sigma_{\dv x}^{-1} + \dv H_{\bar{\mc S}}^\dagger \dv\Sigma_{\dv n_{\bar{\mc S}}}^{-1} \left( \dv I - \dv\Lambda_{\bar{\mc S}} \dv\Sigma_{\dv n_{\bar{\mc S}}}^{-1}  \right) \dv H_{\bar{\mc S}} \right)^{-1} \right| \nonumber\\
& =  \log\left|\dv I + \dv\Sigma_{\dv x} \dv H_{\bar{\mc S}}^\dagger \dv\Sigma_{\dv n_{\bar{\mc S}}}^{-1} \left( \dv I - \dv\Lambda_{\bar{\mc S}} \dv\Sigma_{\dv n_{\bar{\mc S}}}^{-1}  \right) \dv H_{\bar{\mc S}} \right| - I(\dv X; \dv Y_0) \nonumber\\
&= \log\left|\dv I + \dv\Sigma_{\dv x} \dv H_{\bar{\mc S}}^\dagger \dv\Sigma_{\dv n_{\bar{\mc S}}}^{-1} \left( \dv I - \dv\Lambda_{\bar{\mc S}} \dv\Sigma_{\dv n_{\bar{\mc S}}}^{-1}  \right) \dv H_{\bar{\mc S}} \right| - \log\left|\dv I + \dv\Sigma_{\dv x} \dv H^{\dagger}_0 \dv\Sigma^{-1}_0\dv H_0\right|.
\label{equation-Gausss-CEO-second-inequality-final-form}
\end{align}

\noindent Finally, substituting in~\eqref{rates-exponent-region-vector-Gaussian-model-in-terms-of-auxiliaries} using~\eqref{equation-Gausss-CEO-first-inequality} and~\eqref{equation-Gausss-CEO-second-inequality-final-form} we get~\eqref{rate-exponent-region-gaussian-hypothesis-testing-against-conditional-independence}. The proof of the converse terminates by taking the union over all matrices $\dv\Omega_k$ and observing that those satisfy $\dv\Omega_k = \sum_{q\in \mathcal{Q}}p(q) \dv\Omega_{k,q} \preceq\mathbf{\Sigma}_{k}^{-1}$ since $\mathbf{0} \preceq \dv\Omega_{k,q} \preceq\mathbf{\Sigma}_{k}^{-1}$ for all $k \in \mc K$.

%.......................................................................................
\subsubsection{Direct part}
%.......................................................................................

The proof of the direct part follows by evaluating the region described by~\eqref{rates-exponent-region-vector-Gaussian-model-in-terms-of-auxiliaries} and~\eqref{joint-distribution-rates-exponent-region-vector-Gaussian-model-in-terms-of-auxiliaries} using Gaussian test channels and no time-sharing. Specifically, we set $Q=\emptyset$ and
\begin{subequations}
\begin{align}
\dv U_k &= \dv Y_k + \dv V_k \\
&= \dv H_k \dv X + \dv Z_k + \dv V_k 
\end{align}
\label{optimal-test-channel-achievability-vector-Gaussian-case}
\end{subequations}
where the noise $\dv V_k$ is zero-mean Gaussian with covariance matrix 
\begin{equation}
\dv\Gamma_k = \Big[\big(\dv I - \dv\Omega_k\dv\Sigma_k\big)^{-1}-\dv I\Big]^{-1} \dv\Sigma_k
% = {\dv\Sigma_k}^{1/2}(\dv\Omega_k\dv\Sigma_k -\dv I){\dv\Sigma_k}^{1/2}
\label{covariance-matrix-compression-noise-optimal-test-channel-achievability-vector-Gaussian-case}
\end{equation}
 for some matrix $\dv\Omega_k$ such that  $\dv 0 \preceq \dv\Omega_k \preceq \dv\Sigma_k^{-1}$; and is independent of $\dv Y_k$ and of other noises.

\noindent Specifically, using such choice $(\dv U_1,\hdots,\dv U_K,\dv X, \dv Y_0,\dv Y_1,\hdots,\dv Y_K)$ is jointly Gaussian and we get

\begin{align}
I(\dv Y_k;\dv U_k|\dv X,\dv Y_0) &\stackrel{(a)}{=} h(\dv Z_k + \dv V_k|\dv X,\dv Y_0) - h(\dv V_k|\dv X,\dv Y_0,\dv Y_k) \\
&\stackrel{(b)}{=} h(\dv Z_k + \dv V_k|\dv X,\dv Y_0,\dv Z_0) - h(\dv V_k|\dv X,\dv Y_0,\dv Y_k,\dv Z_0, \dv Z_k) \\
&\stackrel{(c)}{=} h(\dv Z_k + \dv V_k|\dv X,\dv Z_0) - h(\dv V_k|\dv X,\dv Z_0, \dv Z_k) \\
&\stackrel{(d)}{=} h(\dv Z_k + \dv V_k|\dv Z_0) - h(\dv V_k|\dv Z_0, \dv Z_k) \\
&\stackrel{(e)}{=} h(\dv Z_k + \dv V_k|\dv Z_0) - h(\dv V_k) \\
&= \log\left|({\pi}e)\mathrm{mmse}(\dv Z_k + \dv V_k|\dv Z_0)\right| - \log\left|({\pi}e)\dv\Gamma_k\right|\\
&= \log\left|\dv\Sigma_k+\dv\Gamma_k\right| - \log\left|\dv\Gamma_k\right|\\
&\stackrel{(f)}{=} -  \log|\dv I- \dv\Omega_k \dv\Sigma_k| 
\label{proof-rates-exponent-region-vector-Gaussian-HT-model-first-term}
\end{align}
where: $(a)$ follows by substituting using~\eqref{optimal-test-channel-achievability-vector-Gaussian-case}; $(b)$ and $(c)$ hold using~\eqref{mimo-gaussian-ht-model-2} and~\eqref{distributions-under-null-hypothesis-ht-against-conditional-independence-vector-Gaussian}; $(d)$ holds since $(\dv Z_0, \dv Z_k, \dv V_k)$ is independent of $\dv X$; $(e)$ holds since the noise $\dv V_k$ is independent of $(\dv Z_0, \dv Z_k)$; $(f)$ follows by substituting using~\eqref{covariance-matrix-compression-noise-optimal-test-channel-achievability-vector-Gaussian-case}. 

\noindent Similarly, for given $\mc S \subseteq \mc K$ we have
\begin{align}
I(\dv U_{\mc S^c};\dv X|\dv Y_0) &\stackrel{(a)}{=} h(\dv U_{\mc S^c}|\dv Y_0) - h(\dv Z_{\mc S^c}+\dv V_{\mc S^c}|\dv X,\dv Y_0) \\
&\stackrel{(b)}{=} h(\dv U_{\mc S^c}|\dv Y_0) - h(\dv Z_{\mc S^c}+\dv V_{\mc S^c}|\dv X,\dv Y_0,\dv Z_0) \\
&\stackrel{(c)}{=} h(\dv U_{\mc S^c}|\dv Y_0) - h(\dv Z_{\mc S^c}+\dv V_{\mc S^c}|\dv Z_0) \\
&\stackrel{(b)}{=} h(\dv U_{\mc S^c}|\dv Y_0) -  \log\left|({\pi}e)(\dv\Sigma_{\mc S^c}+\dv\Gamma_{\mc S^c})\right| \\
&\stackrel{(d)}{=} h(\dv Y_{\mc S^c} + \dv V_{\mc S^c} |\dv Y_0) -  \log\left|({\pi}e)(\dv\Sigma_{\mc S^c}+\dv\Gamma_{\mc S^c})\right| \\
&\stackrel{(e)}{=} \log\left|\dv I + \dv\Sigma_{\dv x} \dv H_{\bar{\mc S}}^\dagger \dv\Sigma_{\dv n_{\bar{\mc S}}}^{-1} \left( \dv I - \dv\Lambda_{\bar{\mc S}} \dv\Sigma_{\dv n_{\bar{\mc S}}}^{-1}  \right) \dv H_{\bar{\mc S}} \right| - \log\left|\dv I + \dv\Sigma_{\dv x} \dv H^{\dagger}_0 \dv\Sigma^{-1}_0\dv H_0\right|
\label{proof-rates-exponent-region-vector-Gaussian-HT-model-second-term}
\end{align}
here: $(a)$ follows by substituting using~\eqref{optimal-test-channel-achievability-vector-Gaussian-case}; $(b)$ holds using ~\eqref{mimo-gaussian-ht-model-2}; $(c)$ holds using $\dv Y_0$ is a deterministic function of $(\dv X, \dv Z_0)$ and $(\dv Z_{\mc S^c}, \dv V_{\mc S^c}, \dv Z_0)$ is independent of $\dv X$; $(d)$ follows by substituting using~\eqref{optimal-test-channel-achievability-vector-Gaussian-case}; and $(e)$ holds using~\eqref{covariance-matrix-compression-noise-optimal-test-channel-achievability-vector-Gaussian-case} and straightforward algebra which is omitted here for brevity.

\noindent Finally, substituting in~\eqref{rates-exponent-region-vector-Gaussian-model-in-terms-of-auxiliaries} using~\eqref{proof-rates-exponent-region-vector-Gaussian-HT-model-first-term} and~\eqref{proof-rates-exponent-region-vector-Gaussian-HT-model-second-term} we get~\eqref{rate-exponent-region-gaussian-hypothesis-testing-against-conditional-independence}; and this completes the proof of the direct part of Theorem~\ref{theorem-rate-exponent-region-gaussian-hypothesis-testing-against-conditional-independence}.

%================================================================================================
\renewcommand{\theequation}{D-\arabic{equation}}
\setcounter{equation}{0}  % reset counter
\subsection{Proof of the Inequality~\eqref{upper-bound-EPI-exponent-rate-function-point-to-point-model}}~\label{appendix-proof-upper-bound-EPI-exponent-rate-function-point-to-point-model}

\noindent Since $Y = X + Z$ with $Z$ Gaussian and independent from $X$, invoking the strong entropy power inequality of Courtade~\cite[Theorem 1]{C18} gives
\begin{equation}
e^{2\big[h(Y) - I(X;U)\big]} \geq e^{2\big[h(X) - I(U; Y)\big]} + e^{2h(Z)}.
\label{proof-upper-bound-EPI-exponent-rate-function-point-to-point-model-step1}
\end{equation}

\noindent Thus, we get
\begin{equation}
I(U;X) \leq h(Y) -\frac{1}{2} \log\left(e^{2\big[h(X) - I(U; Y\big]} + e^{2h(Z)}\right).
\label{proof-upper-bound-EPI-exponent-rate-function-point-to-point-model-step2}
\end{equation}

\noindent Using~\eqref{optimal-exponent-rate-function-point-to-point-model}, we have
\begin{align}
E(R) &= \max_{P_{U|Y} \::\: I(U;Y) \: \leq R} \:\: I(U;X) \\
&\stackrel{(a)}{\leq} \max_{P_{U|Y} \::\: I(U;Y) \: \leq R} \:\: h(Y) -\frac{1}{2} \log\left(e^{2\big[h(X) - I(U; Y)\big]} + e^{2h(Z)}\right) \\
&\stackrel{(b)}{\leq} \max_{P_{U|Y} \::\: I(U;Y) \: \leq R} \:\: h(Y) -\frac{1}{2} \log\left(e^{2\big[h(X) - R\big]} + e^{2h(Z)}\right) \\
&= h(Y) -\frac{1}{2} \log\left(e^{2\big[h(X) - R\big]} + e^{2h(Z)}\right) \\
&\stackrel{(c)}{=} \frac{1}{2} \log \left(\frac{N(Y)}{N(X)e^{-2R}+\sigma^2_Z}\right)
\end{align}
where $(a)$ holds by using~\eqref{proof-upper-bound-EPI-exponent-rate-function-point-to-point-model-step2}, $(b)$ holds by using that $I(U;Y) \leq R$ and $(c)$ holds by substituting using~\eqref{definition-entropy-power}.

%================================================================================================
\renewcommand{\theequation}{E-\arabic{equation}}
\setcounter{equation}{0}  % reset counter
\subsection{Proof of Theorem~\ref{theorem-upper-bound-rate-exponent-function-distributed-setting}}~\label{appendix-proof-theorem-upper-bound-rate-exponent-function-distributed-setting}

Recall the result of Theorem~\ref{theorem-rate-exponent-region-hypothesis-testing-DM-case}. Specializing it to the model described by~\eqref{distribution-under-null-hypothesis-ht-against-independence-2} and~\eqref{distribution-under-alternate-hypothesis-ht-against-independence-2}, we get that the region $\mc R_{\text{HT}}$  is given by the union of all  non-negative tuples $(R_1,\ldots,R_K,E)$ that satisfy, for all subsets $\mc S \subseteq \mc K$, 
\begin{equation}
E \leq I(U_{\mc S^c};X|Q) + \sum_{k \in \mc S} \big(R_k - I(Y_k;U_k|X,Q)\big)
\label{single-letter-characterization-rate-exponent-region}
\end{equation}
for some auxiliary random variables $(U_1,\ldots,U_K,Q)$ with distribution $P_{U_{\mc K},Q}(u_{\mc K},q)$ such that 
\begin{equation}
P_{X, Y_{\mc K}, U_{\mc K}, Q} = P_Q  P_{X} \prod_{k=1}^K P_{Y_k|X}\: \prod_{k=1}^{K} P_{U_k|Y_k,Q}.
\label{distribution-single-letter-characterization-rate-exponent-region}
\end{equation}

\noindent Let $\mc S \subseteq \mc K$ be given, and for $k \in \mc K$ define
\begin{equation}
\gamma_k = \frac{1}{\sigma^2_k} \left(1-e^{-2I(U_k;Y_k|X,Q)}\right).
\label{parametrization-proof-upper-bound}
\end{equation}
\noindent Note that for all $k \in \mc K$, we have $0 \leq \gamma_k \leq 1/\sigma^2_k$.

\noindent If $\mc S=\mc K$, it is easy to see that~\eqref{upper-bound-rate-exponent-function-distributed-setting-full-set} follows directly from~\eqref{single-letter-characterization-rate-exponent-region} using the substitution~\eqref{parametrization-proof-upper-bound}. In the rest of the proof, we therefore suppose that $\mc S$ is a strict subset of $\mc S$, i.e., $\mc S \subset \mc K$.
 
\noindent Using~\eqref{distribution-single-letter-characterization-rate-exponent-region}, it is easy to see that $X \mkv Y(\mc S^c) \mkv (Y_{\mc S^c}, U_{\mc S^c})$ forms a Markov chain conditionally given $Q$. That is,
\begin{equation}
X \mkv Y(\mc S^c) \mkv (Y_{\mc S^c}, U_{\mc S^c}) \:|\: Q.
\label{markov-chain-sufficient-statistic}
\end{equation}

\noindent Since $Y(\mc S^c) = X + Z(\mc S^c)$ with $X$ and $Z(\mc S^c)$ being independent conditionally given $Q$ and $Z(\mc Z^c)$ conditionally Gaussian given $Q$, invoking the conditional strong entropy power inequality of Courtade~\cite[Corollary 2]{C18} yields
\begin{equation}
e^{2\big[h(Y(\mc S^c)|Q) - I(X;U_{\mc S^c}|Q)\big]} \geq e^{2\big[h(X|Q) - I(U_{\mc S^c}; Y(\mc S^c)|Q)\big]} + e^{2h(Z(\mc S^c)|Q)}.
\label{proof-upper-bound-step1}
\end{equation}
\noindent Continuing from~\eqref{proof-upper-bound-step1} using that $Q$ is independent from $(X,Y(\mc S^c), Z(\mc S^c))$, we get
\begin{equation}
e^{2\big[h(Y(\mc S^c)) - I(X;U_{\mc S^c}|Q)\big]} \geq \frac{e^{2h(X)}}{e^{2h(Y(\mc S^c))}} e^{2h\left(Y(\mc S^c)|U_{\mc S^c}, Q\right)} +  e^{2h(Z(\mc S^c))}.
\label{proof-upper-bound-step2}
\end{equation}

\noindent The conditional entropy term $h(Y(\mc S^c)|U_{\mc S^c}, Q)$ is given by
\begin{align}
h\left(Y(\mc S^c)|U_{\mc S^c}, Q\right) &= h\left(Y(\mc S^c)|U_{\mc S^c}, X, Q\right) + I\left(X; Y(\mc S^c), U_{\mc S^c} | Q\right) - I\left(X; U_{\mc S^c} | Q\right) \\
&= h\left(Y(\mc S^c)|U_{\mc S^c}, X, Q\right) + I\left(X; Y(\mc S^c)\right) - I\left(X; U_{\mc S^c} | Q\right)
\label{proof-upper-bound-step3}
\end{align}
where the last equality follows using~\eqref{markov-chain-sufficient-statistic}.

\noindent Also, recalling~\eqref{definition-sufficient-statistic} we have
\begin{align}
 e^{2h\left(Y(\mc S^c)|U_{\mc S^c},  X, Q\right)}  &\stackrel{(a)}{=} e^{2h\left(\frac{1}{|\mc S^c|} \sum_{k \in \mc S^c}\frac{\sigma^2_{\mc S^c}}{\sigma^2_k} Y_k |U_{\mc S^c}, X, Q\right)} \\
&\stackrel{(b)}{\geq} \frac{1}{|\mc S^c|^2} \sum_{k \in \mc S^c} \left(\frac{\sigma^2_{\mc S^c}}{\sigma^2_k}\right)^2  e^{2h(Y_k|U_k,X,Q)}
\label{proof-upper-bound-step4}
\end{align}
 where: $(a)$ follows by substituting using~\eqref{definition-sufficient-statistic}, and $(b)$ follows using the entropy power inequality noticing that for  all $\mc A$ such that $k \notin \mc A$ we have $U_k \mkv Y_k \mkv (X, U_{\mc A}) \: |\: Q$ and the random variables $\{Y_k|U_k, X, Q\}$ are independent.

\noindent Furthermore, we have
\begin{align}
h(Y_k|U_k,X,Q) &= h(Y_k|X,Q) + I(U_k; Y_k | X,Q)\\
&= h(Z_k) + I(U_k;Y_k|X,Q)
\label{proof-upper-bound-step5}
\end{align}
where the last equality follows by substituting using $Y_k=X+Z_k$ and the noise $Z_k$ is independent from $(X,Q)$.

\noindent Now, substituting in~\eqref{proof-upper-bound-step2} using~\eqref{proof-upper-bound-step3}, ~\eqref{proof-upper-bound-step4} and~\eqref{proof-upper-bound-step5}, we get
\begin{align}
e^{2\big[h(Y(\mc S^c)) - I(X;U_{\mc S^c}|Q)\big]} \geq e^{2h(Z(\mc S^c))} & + \frac{e^{2h(X)}}{e^{2h(Y(\mc S^c))}} e^{2\big[I(X; Y(\mc S^c)) - I(X; U_{\mc S^c} | Q)\big]} \nonumber\\
& \times  \frac{1}{|\mc S^c|^2} \sum_{k \in \mc S^c} \left(\frac{\sigma^2_{\mc S^c}}{\sigma^2_k}\right)^2  e^{2h(Z_k)} e^{-2I(U_k; Y_k|X,Q)}.
\label{proof-upper-bound-step6}
\end{align}

\noindent Using~\eqref{proof-upper-bound-step6}, we have
\begin{align}
I(U_{\mc S^c};X|Q) &\leq \frac{1}{2} \log \left( e^{-2h(Z(\mc S^c))} \left[ e^{2h(Y(\mc S^c))} - \frac{e^{2h(X)}e^{2I(X; Y(\mc S^c))}}{e^{2h(Y(\mc S^c))}} \frac{1}{|\mc S^c|^2} \sum_{k \in \mc S^c} \left(\frac{\sigma^2_{\mc S^c}}{\sigma^2_k}\right)^2  e^{2h(Z_k)} e^{-2I(U_k; Y_k|X,Q)} \right] \right) \\
& = \frac{1}{2} \log \left( \frac{e^{2h(Y(\mc S^c))}}{e^{2h(Z(\mc S^c))}} -\frac{e^{2h(X)}}{e^{4h(Z(\mc S^c))}} \frac{1}{|\mc S^c|^2} \sum_{k \in \mc S^c} \left(\frac{\sigma^2_{\mc S^c}}{\sigma^2_k}\right)^2  e^{2h(Z_k)} e^{-2I(U_k; Y_k|X,Q)}  \right) \\
&\stackrel{(a)}{=} \frac{1}{2} \log \left( \frac{N(Y(\mc S^c))}{N(Z(\mc S^c))} -  \frac{N(X)}{(N(Z(\mc S^c)))^2} \sum_{k \in \mc S^c} \left(\frac{\sigma^2_{\mc S^c}}{|\mc S^c|}\right)^2   (\frac{1}{\sigma^2_k} -\gamma_k)  \right) \\
&\stackrel{(b)}{=} \frac{1}{2} \log \left(|\mc S^c| \frac{N(Y(\mc S^c))}{\sigma^2_{\mc S^c}} -N(X) \sum_{k \in \mc S^c} (\frac{1}{\sigma^2_k} - \gamma_k)\right)
\label{proof-upper-bound-step7}
\end{align}
where $(a)$ holds by substituting that for a continuous random variable $A$ we have $e^{2h(A)}=2{\pi}e N(A)$ and $Z_k$ is Gaussian with variance $\sigma^2_k$, and $(b)$ holds by noticing that $N(Z(\mc S^c))=\sigma^2_{\mc S^c}/|\mc S^c|$ and using~\eqref{parametrization-proof-upper-bound}.

\noindent Finally, combining~\eqref{single-letter-characterization-rate-exponent-region} and~\eqref{proof-upper-bound-step7} and substituting using~\eqref{parametrization-proof-upper-bound}, we get~\eqref{upper-bound-rate-exponent-function-distributed-setting-strict-subset}; and this completes the proof of the theorem.

%================================================================================================
\section*{Acknowledgment}
\addcontentsline{toc}{section}{Acknowledgment} 

The author would like to thank Aaron Wagner for fruitful discussions about the relation of Theorem 1 to the outer bound of~\cite[Theorem 2]{RW12}. In particular the steps~\eqref{bound-rahman-wagner-first-step}-~\eqref{bound-rahman-wagner-last-step}, as well as the note of Remark~\ref{relation-to-Courtade-Weissman}, are due to him. The author also thanks the anonymous reviewers for various useful comments and suggestions which improved the quality of this paper.

%================================================================================================
\vspace{-0.4cm}

%-------------------------------------------------------------------------------------
\bibliographystyle{IEEEtran}
\bibliography{draft-paper-arxiv-Gaussian-hypothesis-testing}
%\bibliography{IEEEabrv,D:/Onebox/Work/publications-dissemination/mybibfile}
%\bibliography{IEEEabrv,mybibfile}
\end{document}